\documentclass[bib]{ba}

\usepackage{graphicx,hyperref}
\usepackage{bayes}

\usepackage{amsmath,amssymb,amsthm,bm,natbib,algorithm,algpseudocode,graphicx,color,placeins,esint,enumerate}
\newtheorem*{problem}{Problem}
\newtheorem{lemma}{Lemma}
\usepackage{mathtools}
\DeclarePairedDelimiter{\ceil}{\lceil}{\rceil}

\begin{document}

\inserttype[ba0001]{article}
\renewcommand{\thefootnote}{\fnsymbol{footnote}}
\author{N. Friel, A. Mira and C. J. Oates }{
 \fnms{Nial}
 \snm{Friel}
 \footnotemark[1]\ead{nial.friel@udc.ie},
 \fnms{Antonietta}
 \snm{Mira}
 \footnotemark[2]\ead{antonietta.mira@usi.ch},
and
  \fnms{Chris. J.}
  \snm{Oates}
  \footnotemark[3]\ead{c.oates@warwick.ac.uk}
}

\title[Reduced-Variance Estimation with Intractable Likelihoods]{Exploiting Multi-Core Architectures for Reduced-Variance Estimation with Intractable Likelihoods}

\maketitle

\footnotetext[1]{
 School of Mathematical Sciences and Insight: The National Centre for Data Analytics, University College Dublin, Ireland.
 \href{mailto:nial.friel@udc.ie}{nial.friel@udc.ie}
}
\footnotetext[2]{
Interdisciplinary Institute of Data Science and Institute of Finance, University of Lugano, Switzerland.
 \href{mailto:antonietta.mira@usi.ch}{antonietta.mira@usi.ch}
}
\footnotetext[3]{
Department of Statistics, University of Warwick, UK.
 \href{mailto:c.oates@warwick.ac.uk}{c.oates@warwick.ac.uk}
}
\renewcommand{\thefootnote}{\arabic{footnote}}

\begin{abstract}
Many popular statistical models for complex phenomena are intractable, in the sense that the likelihood function cannot easily be evaluated.
Bayesian estimation in this setting remains challenging, with a lack of computational methodology to fully exploit modern processing capabilities.
In this paper we introduce novel control variates for intractable likelihoods that can dramatically reduce the Monte Carlo variance of Bayesian estimators.
We prove that our control variates are well-defined and provide a positive variance reduction.
Furthermore we show how to optimise these control variates for variance reduction.
The methodology is highly parallel and offers a route to exploit multi-core processing architectures that complements recent research in this direction.
Indeed, our work shows that it may not be necessary to parallelise the sampling process itself in order to harness the potential of massively multi-core architectures.
Simulation results presented on the Ising model, exponential random graph models and non-linear stochastic differential equation models support our theoretical findings.

\end{abstract}

\section{Introduction} \label{intro}

Many models of interest are intractable, by which it is understood that the likelihood function $p(\bm{y}|\bm{\theta})$, that describes how data $\bm{y}$ arise from a model parametrised by $\bm{\theta} \in \Theta$, is unavailable in closed form.
The predominant sources of intractability that are encountered in statistical modelling can be classified as follows:

\begin{tabular}{rp{10cm}}
Type I: & The need to compute a normalising constant $\mathfrak{P}(\bm{\theta}) = \int f(\bm{y}';\bm{\theta}) d\bm{y}'$ that depends on parameters $\bm{\theta}$, such that $p(\bm{y}|\bm{\theta}) = f(\bm{y};\bm{\theta}) / \mathfrak{P}(\bm{\theta})$. \\[.25cm]
Type II: & The need to marginalise over a set of latent variables $\bm{x}$, such that $p(\bm{y}|\bm{\theta}) = \int p(\bm{y}|\bm{x},\bm{\theta})p(\bm{x}|\bm{\theta})d\bm{x}$.
\end{tabular}

Bayesian estimation in both of these settings can be extremely challenging as many established computational techniques (e.g. Gibbs sampling and Metropolis-Hastings) 
are incompatible with intractable likelihoods.
This has motivated researchers to propose several approximations to the likelihood function that are tractable \citep[e.g.][]{Marjoram,Moller,Murray,Rue}.
In the other direction, several (exact) Markov chain Monte Carlo (MCMC) algorithms have been proposed that facilitate inference in intractable models \citep[e.g.][]{Beskos,Andrieu,Andrieu2,Lyne}.
However, for MCMC methodology, 
it remains the case that estimator variance can be heavily inflated relative to the tractable case, due to the need to perform auxiliary calculations on extended state spaces in order to address the intractability \citep{Sherlock}. 
Below we elaborate on the two types of intractability and on the related references in the literature that have addressed them.

\paragraph{Type I:}

Intractability arises from the need to compute a parameter-dependent 
normalising constant (sometimes called a partition function).
This paper focuses on the sub-class of Type I intractable models known as Gibbs random fields (GRFs) where data $\bm{y}$ arises from a model of the form
\begin{eqnarray}
\log p(\bm{y}|\bm{\theta}) = \bm{\theta}^T\bm{s}(\bm{y}) - \log\mathfrak{P}(\bm{\theta})
\end{eqnarray}
such that the partition function
\begin{eqnarray}
\mathfrak{P}(\bm{\theta}) = \int \exp(\bm{\theta}^T\bm{s}(\bm{y})) d\bm{y}
\end{eqnarray}
is intractable. 
In a Bayesian context this leads to a ``doubly intractable'' distribution and is the subject of current research in the statistical
community.
Below we survey applications of, and methodology for, models exhibiting this form of intractability:

{\it Example 1 (Spatial statistics):
A GRF-type 
intractability arises in classical spatial statistics where we seek to model the joint distribution of variables $Y_j$ that are subject to local interactions.
The autologistic distribution \citep{bes74} is well-studied model 
for the analysis of binary spatial data defined on a lattice. This model has been applied in diverse 
contexts including ecology \citep{aug:mug96}, the spatial analysis of plant species \citep{huffer:wu,he:zhou:zhu} and dentistry \citep{Bandyopadhyay}.
The canonical Ising model is a special case of the autologistic distribution and is defined on a regular lattice of size $n \times n$, where $j$ is used to index each of the $n \times n$ different lattice locations.
Here the random variable $\bm{Y} \in \{-1,1\}^{n \times n}$ has a probability distribution defined in terms of a single sufficient statistic
\[
 s(\bm{y}) = \sum_{j=1}^{n \times n} \sum_{i\sim j} y_i y_j,
\]
where the notation $i\sim j$ means that the lattice point $i$ is a neighbour of lattice point $j$. 
Interactions are modelled between neighbouring lattice points $i \sim j$, being captured by the energy term $y_iy_j$.
The likelihood for this model takes the form of a GRF where the partition function
\begin{eqnarray}
\mathfrak{P}(\theta) = \sum_{\bm{y}' \in \{-1,1\}^{n \times n}} \exp(\theta s(\bm{y}'))
\end{eqnarray}
involves the summation over $2^{n \times n}$ different possible state vectors $\bm{y}'$.
Typically this summation is infeasible and leads to Type I intractability for all but small values of the lattice size $n$.
}

{\it Example 2 (Social network analysis):
Exponential random graph (ERG) models are widely used in social network analysis \citep[see][and the references therein]{rob:pat:kal:lus07}. 
The ERG model is defined on a random adjacency matrix $\bm{Y} = \{ Y_{ij}: i=1,\dots,n; j=1,\dots,n\}$ of a graph with 
$n$ nodes where $Y_{ij}=1$ if nodes $i$ and $j$ are connected by an edge, and $Y_{ij}=0$ otherwise. An edge connecting a node to itself is not permitted so $Y_{ii}=0$. 
The edges in an ERG may be undirected, whereby $Y_{ij}=Y_{ji}$, or directed, whereby a directed 
edge from node $i$ to node $j$ is not necessarily reciprocated. 
Write $\mathcal{G}(n)$ for the set of all permitted graphs on $n$ vertices.
The likelihood of an observed graph $\bm{y}$ is modelled in terms of a collection of sufficient statistics 
$\bm{s}(\bm{y}) = (s_1(\bm{y}),\dots,s_k(\bm{y}))$ and corresponding parameters $\bm{\theta}= (\theta_1,\dots,\theta_k)$.
For example, typical statistics include $s_1(y) = \sum_{i<j}y_{ij}$ and $s_2(y) = \sum_{i<j<k}y_{ik}y_{jk}$ that encode, respectively, the observed number of edges and two-stars, that is, the number of configurations of pairs of edges 
that share a common node. It is also possible to consider statistics that count the number of configuration of $k$  
edges that share a node in common, for $k>2$. 
The likelihood takes the form of a GRF where the partition function
\begin{eqnarray}
\mathfrak{P}(\bm{\theta}) = \sum_{\bm{y}' \in \mathcal{G}(n)} \exp(\bm{\theta}^T \bm{s}(\bm{y}'))
\end{eqnarray}
involves the summation over $|\mathcal{G}(n)| = O(2^{n \times n})$ possible different graphs and leads to Type I intractability for all but small values of the number $n$ of vertices.
}

The dependence of the partition function $\mathfrak{P}(\bm{\theta})$ on $\bm{\theta}$ leads to difficulties in inferring this parameter. 
An early attempt to circumvent this difficulty is the pseudolikelihood approach of \cite{Besag}, which in turn has been generalised to composite likelihood approximations, see for example \cite{Davison}. 
An alternative class of inferential approaches results from realising that, although one cannot evaluate the likelihood function, it is possible to sample pseudo-data from the generative model, so-called ``forward simulation''.
The Monte Carlo MLE approach of \cite{gey:tho92} exploits forward simulation to allow maximum likelihood estimation. 
From a Bayesian perspective, simulating from the likelihood has also played an influential role in several approaches,
for example, the auxiliary variable method of \cite{Moller}, that was subsequently extended by \cite{Murray} to the exchange algorithm. 
The exchange algorithm avoids the need to directly evaluate the partition function by considering an augmented target distribution $p(\bm{\theta},\bm{\theta}',\bm{y}'|\bm{y})$ 
that includes a second copy $\bm{\theta}'$ of the parameter vector and forward-simulated pseudo-data $\bm{y}'$ drawn from the likelihood function $p(\bm{y}'|\bm{\theta}')$, defined in such a way that the Markov chain transition kernel for the parameter vector $\bm{\theta}$ of interest involves partition functions for the current and proposed values of $\bm{\theta}$ that cancel in the numerator and denominator of the Metropolis-Hastings ratio, thus circumventing the Type I intractability issue (see Alg. \ref{exchange}) at the expense of increased Monte Carlo variance.

An emerging research direction is the construction of approximate Monte Carlo algorithms, providing convergence guarantees, in situations where it is
expensive or impossible to calculate the likelihood. This is particularly pressing in cases where the exchange algorithm is applied, since forward simulating from Gibbs random fields is challenging.
Perfect sampling is often prohibitively expensive or impossible to carry out, and in this case \cite{Everitt} has provided convergence results for the case where one uses the final draw from a
Gibbs sampler targeting the likelihood as an approximate realisation. In a similar vein, \cite{Alquier} and \cite{pillai} develop convergence results for approximate MCMC algorithms resulting
from approximating the transition kernel due to the intractability of the likelihood function. It is worth noting that several authors have used this type of approach to develop approximate
algorithms for large datasets by using subsets of the data to approximate the likelihood, \citep{welling,anh,Korattikara}.

\begin{algorithm}[t!]
\caption{Exchange algorithm for Type I intractability \citep{Murray}}\label{exchange}
\begin{algorithmic}[1]
\State Initialise $\bm{\theta}^{(0)}$, ${\bm{\theta}^{(0)}}'$, ${\bm{y}^{(0)}}'$.
\For{$i = 1,\dots,I$}
\State Obtain $\bm{\theta}' \sim h(\bm{\theta}'|\bm{\theta}^{(i-1)})$.
\State Obtain $\bm{y}' \sim p(\bm{y}'|\bm{\theta}')$.
\State Exchange $\bm{\theta}' \mapsto \bm{\theta}^{(i)}$, $\bm{\theta}^{(i-1)} \mapsto {\bm{\theta}^{(i)}}'$, $\bm{y}' \mapsto {\bm{y}^{(i)}}'$ with probability
\begin{eqnarray}
\alpha = \min \left\{ 1 , \frac{p(\bm{y}|\bm{\theta}')}{p(\bm{y}'|\bm{\theta}')} \frac{p(\bm{y}'|\bm{\theta}^{(i-1)})}{p(\bm{y}|\bm{\theta}^{(i-1)})} \frac{p(\bm{\theta'})}{p(\bm{\theta})} \frac{h(\bm{\theta}^{(i-1)|\bm{\theta}'})}{h(\bm{\theta}'|\bm{\theta}^{(i-1)})} \right\},
\end{eqnarray}
\; \; \; otherwise set $\bm{\theta}^{(i-1)} \mapsto \bm{\theta}^{(i)}$, ${\bm{\theta}^{(i-1)}}' \mapsto {\bm{\theta}^{(i)}}'$, ${\bm{y}^{(i-1)}}' \mapsto {\bm{y}^{(i)}}'$.
\EndFor
\end{algorithmic}
\end{algorithm}

\paragraph{Type II:}

Intractability arises from the need to marginalise over latent variables $\bm{x}$ such that the marginal likelihood 
\begin{eqnarray}
p(\bm{y}|\bm{\theta}) = \int p(\bm{y}|\bm{x},\bm{\theta}) p(\bm{x}|\bm{\theta}) d\bm{x} \label{ml}
\end{eqnarray}
is unavailable in closed form.
Such problems arise frequently in applied statistics and examples include inference for the parameters of spatio-temporal models \citep{Rue,Lyne}, regression models with random effects \citep{Fahrmeir}, time-series models \citep{West}, and selection between competing models based on Bayes factors \citep[e.g.][]{Caimo4,Armond}.
Below we provide examples of, and survey methodology for, models exhibiting this form of intractability:

{\it Example 3 (Hidden Markov model):
Applications of hidden Markov models abound in many areas, including finance, economics and biology. See \citep{cappe} for a detailed analysis of this general area.
In a hidden Markov model, 
the parameters $\bm{\theta}$ that specify a Markov chain
\begin{eqnarray}
\bm{x}_{n+1} \sim p(\bm{x}_{n+1}|\bm{x}_n,\bm{\theta})
\end{eqnarray}
may be of interest, whilst the latent sample path $\{\bm{x}_n\}_{n=0}^N$ of the Markov chain that gives rise to observations $\bm{y}_{n} \sim p(\bm{y}_{n}|\bm{x}_n)$ may not be of interest and must be marginalised.
Even in discrete cases where $\bm{x}_n \in \mathcal{X}$ for a finite state space $\mathcal{X}$, the number of possible samples paths $\{\bm{x}_n\}_{n=0}^N$ grows exponentially in $N$ and this renders the marginalisation
\begin{eqnarray}
p(\{\bm{y}_n\}_{n=0}^N|\bm{\theta}) = \sum_{\bm{x}_0,\dots,\bm{x}_N \in \mathcal{X}} p(\{\bm{y}_n\}_{n=0}^N|\{\bm{x}_n\}_{n=0}^N,\bm{\theta}) p(\{\bm{x}_n\}_{n=0}^N|\bm{\theta})
\end{eqnarray}
corresponding to Eqn. \ref{ml} computationally intractable.
}

{\it Example 4 (Stochastic differential equations):
Stochastic differential equations (SDEs) are widely used in several fields including biology \citep{Wilkinson} and finance \citep{Lamberton}. See \cite{oksendal} for an excellent introduction to SDEs, including a focus on several application areas.
A general stochastic diffusion is defined as
\begin{eqnarray}
d\bm{X}(t) =  \bm{\alpha}(\bm{X}(t);\bm{\theta}) dt + \bm{\beta}^{1/2}(\bm{X}(t);\bm{\theta}) d\bm{W}(t), \; \; \; \; \; \bm{X}(0) = \bm{X}_0,  \label{general sde}
\end{eqnarray}
where $\bm{X}(t)$ is a stochastic process taking values in $\mathbb{R}^d$, $\bm{\alpha} : \mathbb{R}^d \times \Theta \rightarrow \mathbb{R}^d$ is a drift function, $\bm{\beta} : \mathbb{R}^d \times \Theta \rightarrow \mathbb{R}^d \times \mathbb{R}^d$ is a diffusion function, $\bm{W}(t)$ is a $d$-dimensional Weiner process, $\bm{\theta} \in \Theta$ are unknown model parameters and $\bm{X}_0 \in \mathbb{R}^d$ is an initial state (assumed known here).
For general SDEs, an analytic form for the distribution of sample paths is unavailable.
An excellent review of approximate likelihood methods for SDEs is provided in \cite{Fuchs}.
To facilitate inference here, a popular approach is to introduce a fine discretisation $t_1,\dots,t_T$ of time with mesh size $\delta t$.
Write $\bm{X}_i = \bm{X}(t_i)$.
The Euler-Maruyama approximation to the SDE likelihood is then given by
\begin{eqnarray}
p(\bm{X}|\bm{\theta}) \propto \prod_{i=2}^T \psi(\bm{X}_i|\bm{X}_{i-1} + \bm{\alpha}_{i} \delta t, \bm{\beta}_{i} \delta t) 
\end{eqnarray}
where $\psi(\cdot|\bm{\mu},\bm{\Sigma})$ is the probability density function for a Gaussian random variable with mean $\bm{\mu}$ and covariance $\bm{\Sigma}$ and where we have used the shorthand $\bm{\alpha}_i = \bm{\alpha}(\bm{X}_{i-1};\bm{\theta})$ and $\bm{\beta}_i = \bm{\beta}(\bm{X}_{i-1};\bm{\theta})$.
We partition $\bm{X} = [\bm{X}^o \; \bm{X}^u]$ such that $\bm{y} = \bm{X}^o$ are observed (for simplicity here without noise) and $\bm{x} = \bm{X}^u$ are unobserved.
This is essentially a hidden Markov model with a continuous latent state and therefore exhibits Type II intractability, since to draw inferences on $\bm{\theta}$ it is required to marginalise the unobserved variables $\bm{X}^u$.
}

A popular contemporary approach to inference under Type II intractability is the pseudo-marginal MCMC of \cite{Andrieu}, that replaces the marginal likelihood $p(\bm{y}|\bm{\theta})$ in the Metropolis-Hastings acceptance ratio with an unbiased estimate that can either be obtained by forward-simulation from $p(\bm{x}|\bm{\theta})$, or using importance sampling techniques.
The pseudo-marginal MCMC typically leads to reduced efficiency relative to the (unavailable) marginal algorithm, but improved efficiency relative to a Markov chain constructed on the extended space $(\bm{\theta},\bm{x})$ \citep{Sherlock}.
When combined with particle MCMC \citep{Andrieu2}, the pseudo-marginal algorithm represents a popular technique to deal with general forms of Type II intractability.
Within specific model classes it may be possible to design specialised approaches to estimation; for example \cite{Kou} and \cite{Beskos2} both present sophisticated schemes for parameter inference in discretely observed stochastic differential equation models.
Other attempts to address Type II intractability include the popular approximation scheme of \cite{Rue} and the references therein.
Such schemes trade exactness of computation for substantial reduction in computational effort, but many questions surround the extent of approximation error \citep[e.g.][]{Lindgren}.

In summary, applications involving statistical models with both types of intractability are widespread in the literature.
Moreover, as detailed above, statistical methodology to overcome both types of intractability is at the frontier of research in computational statistics. Indeed one might anticipate that even wider 
applicability will result as these methods disseminate in the scientific community, whereby hitherto intractable statistical models will be amenable to statistical inference.

\paragraph{Outline of the paper:}

The present contribution addresses the problem of estimating posterior expectations via MCMC when data arise from an intractable likelihood:
\begin{problem}
Estimate the posterior expectation $\mu = \mathbb{E}_{\bm{\theta}|\bm{y}}[g(\bm{\theta})]$ for some known function $g: \Theta \rightarrow \mathbb{R}$, where data $\bm{y}$ arise from an intractable likelihood of either Type I or II.
\end{problem}

Our focus is on the use of control variates for the reduction of Monte Carlo variance \citep{Glasserman}.
The basic idea behind control variate schemes in Bayesian computation is that a modified function $\tilde{g}(\bm{\theta}) = g(\bm{\theta}) + \phi_1 h_1(\bm{\theta}) + \dots + \phi_m h_m(\bm{\theta})$ is constructed such that $\tilde{g}(\bm{\theta})$ has the same posterior expectation but a reduced posterior variance compared to $g(\bm{\theta})$.
This can occur when (i) each of the $h_i(\bm{\theta})$ have zero posterior expectation, (ii) the collection $[h_1(\bm{\theta}),\dots,h_m(\bm{\theta})]$ has strong posterior canonical correlation with the target $g(\bm{\theta})$ and (iii) the coefficients $\phi_1,\dots,\phi_m$ are chosen appropriately. 
Recently \cite{Mira} proposed to use the score vector $\bm{u}(\bm{\theta}|\bm{y}) := \nabla_{\bm{\theta}} \log p(\bm{\theta}|\bm{y})$ as the basis for a set of control variates, since this can be guaranteed to have zero expectation under mild boundary conditions (described below).
There it was shown that these score-based control variates can significantly reduce Monte Carlo variance, sometimes dramatically.
Indeed, the methodology was named ``zero variance'' (ZV) by \cite{Mira},
following \cite{Assaraf},
since in several special cases the score has perfect canonical correlation with the target, generating an estimate that has zero sampling variance.
Further support for the use of the score as a control variate was provided in \cite{Papamarkou,Oates}, who demonstrated that the approach fits naturally within 
Hamiltonian-type and Langevin-type MCMC schemes that themselves make use of the score, requiring essentially no additional computational effort.
It would therefore be extremely desirable to design control variates for intractable likelihoods, where sampling variance is acutely problematic.
However, for Bayesian inference with intractable likelihoods, the score is unavailable as it requires the derivative of unknown quantities.
Our work is motivated by overcoming this impasse.

The main contribution of this paper is to introduce a stochastic approximation to ZV control variates, called ``reduced-variance'' (RV) 
control variates, that can be computed for intractable likelihoods of both Type I and II.
Specifically we study the effect of replacing the true score function $\bm{u}(\bm{\theta}|\bm{y})$ for the intractable models in the ZV methodology with an unbiased estimate $\hat{\bm{u}}(\bm{\theta}|\bm{y})$ that can be obtained via repeated forward-simulation.
Importantly, these forward-simulations can be performed in parallel, offering the opportunity to exploit modern multi-core processing architectures \citep{Suchard,Lee} in a straight-forward manner that directly complements (and is compatible with) related research efforts for parallelisation of MCMC methodology \citep{Alquier,Angelino,Bardenet,Calderhead,Korattikara,Maclaurin}.

From a theoretical perspective, we prove that RV control variates are well-defined and provide a positive variance reduction.
Furthermore we propose default tuning parameters that are proven to maximise variance reduction and prove that the optimal 
estimator for serial computation requires essentially the same computational effort as the state-of-the-art estimate obtained under either the exchange algorithm or the pseudo-marginal algorithm.
These results are orthogonal to recent work by \cite{Doucet} and \cite{Sherlock} that deals with implementation of MCMC samplers themselves.
Empirical results presented on the Ising model, exponential random graphs and nonlinear stochastic differential equations support our theoretical findings.

\section{Methods}

\subsection{Control variates and intractable likelihoods}

Our presentation of control variate methodology below focuses on the problem of evaluating posterior expectations, but the methodology itself applies more broadly.
In this restricted setting, control variates can be employed when the aim is to estimate, with high precision, the posterior expectation $\mu = \mathbb{E}_{\bm{\theta}|\bm{y}}[g(\bm{\theta})]$ of a (real-valued) function $g(\bm{\theta})$ of an unknown parameter $\bm{\theta}$.
In this paper we focus on a real-valued random parameter $\bm{\theta} \in \Theta \subseteq \mathbb{R}^d$.
The generic control variate principle relies on constructing an auxiliary function $\tilde{g}(\boldsymbol{\theta})=g(\boldsymbol{\theta})+ h(\boldsymbol{\theta})$ 
where $\mathbb{E}_{\bm{\theta}|\bm{y}}[h(\bm{\theta})] = 0$ and so
$\mathbb{E}_{\bm{\theta}|\bm{y}}[\tilde{g}(\boldsymbol{\theta})] = \mathbb{E}_{\bm{\theta}|\bm{y}}[g(\boldsymbol{\theta})]$.
In many cases it is possible to choose $h(\bm{\theta})$ such that the variance $\mathbb{V}_{\bm{\theta}|\bm{y}}[\tilde{g}(\bm{\theta})] < \mathbb{V}_{\bm{\theta}|\bm{y}}[g(\bm{\theta})]$, leading to a Monte Carlo estimator with strictly smaller variance:
\begin{eqnarray}
\hat{\mu} := \frac{1}{n} \sum_{i=1}^n \tilde{g}(\bm{\theta}^{(i)}),
\end{eqnarray}
where $\bm{\theta}^{(1)}, \dots,\bm{\theta}^{(n)}$ are independent samples from $p(\bm{\theta}|\bm{y})$.
Intuitively, greater variance reduction can occur when $h(\bm{\theta})$ is negatively correlated with $g(\bm{\theta})$ in the posterior, since much of the randomness ``cancels out'' in the auxiliary function $\tilde{g}(\bm{\theta})$.

In classical literature the function $h(\bm{\theta})$ is often formed as a sum $\phi_1 h_1(\bm{\theta}) + \dots \phi_m h_m(\bm{\theta})$ where the $h_i(\bm{\theta})$ each have zero posterior expectation 
(under the target)  
and are known as control variates, whilst $\phi_i$ are coefficients that must be specified \citep{Glasserman}.
Alternative constructions also exist \citep[e.g. ratio control variates;][]{Evans} but here we focus only on control variates with an additive structure.
For estimation based on Markov chains, \cite{Andradottir} proposed control variates for discrete state spaces. 
Later \cite{Mira2} extended the approach of \cite{Assaraf}
observing that the optimal choice of $h(\bm{\theta})$ is intimately associated with the solution of the Poisson equation $h(\bm{\theta}) = \mathbb{E}_{\bm{\theta}|\bm{y}}[g(\bm{\theta})] - g(\bm{\theta})$ and proposing to solve this equation numerically.
Further work to construct control variates for Markov chains includes \cite{Hammer} for Metropolis-Hastings samplers and \cite{Dellaportas} for Gibbs samplers.

In this paper we consider the particularly elegant class of control variates that are expressed as functions of the score vector $\bm{u}(\bm{\theta}|\bm{y})$ of the log-posterior density.
\cite{Mira} proposed the ZV control variates
\begin{eqnarray}
h(\bm{\theta}|\bm{y}) = \Delta_{\bm{\theta}}[P(\bm{\theta})] + \nabla_{\bm{\theta}}[P(\bm{\theta})] \cdot \bm{u}(\bm{\theta}|\bm{y})
\label{ZV h}
\end{eqnarray}
where $\nabla_{\bm{\theta}} = [\partial/\partial_{\theta_1},\dots,\partial/\partial_{\theta_d}]^T$ is the gradient operator, $\Delta_{\bm{\theta}} = (\partial^2/\partial_{\theta_1}^2 + \dots + \partial^2/\partial_{\theta_d}^2)$ is the Laplacian operator and the ``trial function'' $P(\bm{\theta})$ belongs to the family $\mathcal{P}$ of polynomials in $\bm{\theta}$.
In this paper we adopt the convention that both $\bm{\theta}$ and $\bm{u}(\bm{\theta}|\bm{y})$ are $d \times 1$ vectors.
\cite{Mira} showed, in particular, that any posterior density $p(\bm{\theta}|\bm{y})$ approximating a Gaussian forms 
a suitable candidate for implementing the ZV scheme.
The ZV approach has recently been extended to encompass non-parametric trial functions $P(\bm{\theta})$.
\cite{Oates2} proves that the associated estimators posses superior convergence rates relative to estimation that does not use control variates.
A consequence of this latter approach is that large variance reductions can be achieved outside of the Gaussian setting.
For a comprehensive review of the ZV methodology see
\cite{Papamarkou2}.
Unfortunately ZV methods are not directly compatible with intractable likelihoods:

\paragraph{Type I:}
A naive application of ZV methods to GRFs with Type I intractability would require the score function, that is obtained by differentiating
\begin{eqnarray}
\log p(\bm{\theta}|\bm{y}) = \bm{\theta}^T\bm{s}(\bm{y}) - \log \mathfrak{P}(\bm{\theta}) + \log p(\bm{\theta}) + C,
\end{eqnarray}
where $C$ is a constant in $\bm{\theta}$, to obtain
\begin{eqnarray}
\bm{u}(\bm{\theta}|\bm{y}) = \bm{s}(\bm{y}) - \nabla_{\bm{\theta}} \log \mathfrak{P}(\bm{\theta}) + \nabla_{\bm{\theta}} \log p(\bm{\theta}). \label{deriv}
\end{eqnarray}
It is clear that Eqn. \ref{deriv} will not have a closed-form when the partition function $\mathfrak{P}(\bm{\theta})$ is intractable.
In the sections below we demonstrate how forward-simulation can be used to approximate $\nabla_{\bm{\theta}} \log \mathfrak{P}(\bm{\theta})$ and then leverage this fact to reduce Monte Carlo variance.

\paragraph{Type II:}
Similarly, a naive application of ZV within Type II intractable likelihood problems would require 
the evaluation of
the score function
\begin{eqnarray}
\bm{u}(\bm{\theta}|\bm{y}) = \nabla_{\bm{\theta}} \log \int p(\bm{y},\bm{x}|\bm{\theta}) p(\bm{x}|\bm{\theta}) d\bm{x} + \nabla_{\bm{\theta}} \log p(\bm{\theta}). \label{deriv2}
\end{eqnarray}
It is clear that Eqn. \ref{deriv2} will not have a closed-form when the integral over the latent variable $\bm{x}$ is intractable.
In the sections below we demonstrate how forward-simulation can be used to approximate $\nabla_{\bm{\theta}} \log \int p(\bm{y},\bm{x}|\bm{\theta}) p(\bm{x}|\bm{\theta}) d\bm{x}$, before again leveraging this fact to reduce Monte Carlo variance.

\subsection{Unbiased estimation of the score}

Our approach relies on the ability to construct an unbiased estimator for the score function in both Type I and Type II intractable models.

\paragraph{Type I:}
An unbiased estimator for $\bm{u}(\bm{\theta}|\bm{y})$, that can be computed for Type I models of GRF form, is constructed by noting that
\begin{eqnarray}
\nabla_{\bm{\theta}} \log\mathfrak{P}(\bm{\theta}) & = & \frac{1}{\mathfrak{P}(\bm{\theta})} \nabla_{\bm{\theta}} \mathfrak{P}(\bm{\theta})  \\
& = & \frac{1}{\mathfrak{P}(\bm{\theta})} \nabla_{\bm{\theta}} \int \exp(\bm{\theta}^T\bm{s}(\bm{y})) d\bm{y} \\
& = & \frac{1}{\mathfrak{P}(\bm{\theta})} \int \bm{s}(\bm{y}) \exp(\bm{\theta}^T\bm{s}(\bm{y})) d\bm{y} \\
& = & \mathbb{E}_{\bm{Y}|\bm{\theta}}[\bm{s}(\bm{Y})], \label{final typeI}
\end{eqnarray}
where we have assumed regularity conditions that permit the interchange of derivative and integral operators (including that the domain of $\bm{Y}$ does not depend on $\bm{\theta}$).
Specifically, 
combining Eqns. \ref{deriv} and \ref{final typeI}
we estimate the score function by exploiting multiple forward-simulations
\begin{eqnarray}
\hat{\bm{u}}(\bm{\theta}|\bm{y}) := \bm{s}(\bm{y}) - \left[\frac{1}{K} \sum_{k=1}^K \bm{s}(\bm{Y}_k)\right] + \nabla_{\bm{\theta}} \log p(\bm{\theta})
\end{eqnarray}
where the $\bm{Y}_1,\dots,\bm{Y}_K$ are independent simulations from the GRF with density $p(\bm{y}|\bm{\theta})$.
Forward-simulation for GRF has previously been leveraged to facilitate estimation \citep[e.g.][]{Potamianos} and can be achieved using, for example, perfect sampling \citep{Propp,Mira3}.
We make two important observations:
Firstly, one realisation $\bm{Y}_1$ must be drawn in any case to perform the exchange algorithm, so that this requires no additional computation.
Secondly, these $K$ simulations can be performed in parallel, enabling the exploitation of multi-core processing architectures.

\paragraph{Type II:}
For intractable models of Type II an alternative approach to construct an unbiased estimate for the score is required.
Specifically, we notice that the score $\bm{u}(\bm{\theta},\bm{x}) := \nabla_{\bm{\theta}} \log p(\bm{\theta},\bm{x}|\bm{y})$ of the extended
posterior is typically available in closed form and this can be leveraged as follows:
\begin{eqnarray}
\bm{u}(\bm{\theta} | \bm{y}) =
\nabla_{\bm{\theta}} \log p(\bm{\theta}|\bm{y}) & = & \frac{\nabla_{\bm{\theta}} p(\bm{\theta}|\bm{y})}{p(\bm{\theta}|\bm{y})} \\
& = & \frac{1}{p(\bm{\theta}|\bm{y})} \nabla_{\bm{\theta}} \int p(\bm{\theta},\bm{x}|\bm{y}) d\bm{x} \\
& = & \int \frac{[\nabla_{\bm{\theta}} p(\bm{\theta},\bm{x}|\bm{y})]}{p(\bm{\theta},\bm{x}|\bm{y})} \frac{p(\bm{\theta} , \bm{x}|\bm{y})}{p(\bm{\theta}|\bm{y})} d\bm{x} \\
& = & \int [\nabla_{\bm{\theta}} \log p(\bm{\theta},\bm{x}|\bm{y})] p(\bm{x}|\bm{\theta},\bm{y}) d\bm{x} = \mathbb{E}_{\bm{X}|\bm{\theta},\bm{y}}[\bm{u}(\bm{\theta},\bm{X})]\; \; \; \; \; \; \label{FI}
\end{eqnarray}
where again we have assumed regularity conditions that allow us to interchange the integral and the derivative operators.
We therefore have a simulation-based estimator
\begin{eqnarray}
\hat{\bm{u}}(\bm{\theta}|\bm{y}) := \frac{1}{K} \sum_{k=1}^K \bm{u}(\bm{\theta},\bm{X}_k)
\end{eqnarray}
where the $\bm{X}_1,\dots,\bm{X}_K$ are independent simulations from the posterior conditional $p(\bm{x}|\bm{\theta},\bm{y})$.
\citep[Eqn. \ref{FI} is sometimes called ``Fisher's identity'';][]{Nemeth}.
We note that it is straight-forward to implement pseudo-marginal MCMC in such a way that samples $\bm{X}_i$ are obtained
as a by-product, so that estimation of the score requires no additional computation.

\subsection{Reduced-variance control variates}

This paper advocates constructing control variates using an unbiased estimator for the score as follows:
\begin{eqnarray}
\hat{h}(\bm{\theta}|\bm{y}) := \Delta_{\bm{\theta}}[P(\bm{\theta})] + \nabla_{\bm{\theta}}[P(\bm{\theta})] \cdot \hat{\bm{u}}(\bm{\theta}|\bm{y}), \label{h}
\end{eqnarray}
where again $P \in \mathcal{P}$ is a polynomial trial function.
The coefficients $\bm{\phi}$ of this polynomial $P(\bm{\theta})$ must be specified and we will also write $P(\bm{\theta}|\bm{\phi})$ to emphasise this point.
These will be referred to as ``reduced-variance'' control variates from the fact that Eqn. \ref{h} is a stochastic approximation to the ZV control variates 
in Eqn. \ref{ZV h}
and can therefore be expected to have similar properties.
Pseudocode is provided in Alg. \ref{exchange2}.

\begin{algorithm}[t!]
\caption{Reduced-variance estimation for intractable likelihoods}\label{exchange2}
\begin{algorithmic}[1]
\State Obtain $\bm{\theta}^{(i)} \sim \bm{\theta}|\bm{y}$, $i = 1,\dots,I$ \Comment{using MCMC}
\For{$i = 1,\dots,I$}
\If{Type I}
\State Obtain $\bm{y}^{(i,k)} \sim \bm{Y}|\bm{\theta}^{(i)}$, $k = 1,\dots,K$ \Comment{simulate from the model}
\State Construct an approximation to the score at $\bm{\theta}^{(i)}$:
\begin{eqnarray}
\hat{\bm{u}}^{(i)} = \bm{s}(\bm{y}) - \left[\frac{1}{K} \sum_{k=1}^K \bm{s}(\bm{y}^{(i,k)})\right] + \nabla_{\bm{\theta}} \log p(\bm{\theta}^{(i)})
\end{eqnarray}
\Else{ {\bf if} Type II {\bf then}}
\State Obtain $\bm{x}^{(i,k)} \sim \bm{x}|\bm{\theta}^{(i)},\bm{y}$, $k = 1,\dots,K$ \Comment{simulate from the posterior}
\State Construct an approximation to the score at $\bm{\theta}^{(i)}$:
\begin{eqnarray}
\hat{\bm{u}}^{(i)} = \frac{1}{K} \sum_{k=1}^K \bm{u}(\bm{\theta},\bm{x}^{(i,k)})
\end{eqnarray}
\EndIf
\EndFor
\State Estimate optimal polynomial coefficients $\bm{\phi}^*$ by $\hat{\bm{\phi}}$ \Comment{see section \ref{phi hat}} 
\For{$i = 1,\dots,I$}
\State Construct the reduced-variance control variates
\begin{eqnarray}
\hat{h}^{(i)} = \Delta_{\bm{\theta}}[P(\bm{\theta}^{(i)}|\hat{\bm{\phi}})] + \nabla_{\bm{\theta}}[P(\bm{\theta}^{(i)}|\hat{\bm{\phi}})] \cdot \hat{\bm{u}}^{(i)}.
\end{eqnarray}
\EndFor
\State Estimate the expectation $\mu$ using
\begin{eqnarray}
\hat{\mu} := \frac{1}{I} \sum_{i=1}^I g(\bm{\theta}^{(i)}) + \hat{h}^{(i)}.
\end{eqnarray}
\end{algorithmic}
\end{algorithm}

For this idea to work it must be the case that the RV
control variates $\hat{h}(\bm{\theta}|\bm{y})$ have zero expectation. This is guaranteed under mild assumptions stated below:

\begin{lemma} \label{conditions}
Assume that $\Theta$ is possibly unbounded, $B_r$ are bounded sets increasing to $\Theta$ and $\lim_{r \rightarrow \infty} \oint_{\partial B_r} 
p(\bm{\theta}|\bm{y}) \nabla P(\bm{\theta}) \cdot \bm{n}(\bm{\theta}) d\bm{\theta} = 0$, where $\bm{n}(\bm{\theta})$ is the outward pointing unit normal field of the boundary $\partial B_r$.
Then, for Type I models, 
$\mathbb{E}_{\bm{\theta},\bm{Y}_1,\dots,\bm{Y}_K|\bm{y}}[\hat{h}(\bm{\theta}|\bm{y})] = 0$, whilst, for Type II models, 
$\mathbb{E}_{\bm{\theta},\bm{X}_1,\dots,\bm{X}_K|\bm{y}}[\hat{h}(\bm{\theta}|\bm{y})] = 0$, so that in both cases $\hat{h}(\bm{\theta}|\bm{y})$ is a well-defined control variate.
\end{lemma}
\begin{proof}
From unbiasedness of $\hat{\bm{u}}(\bm{\theta}|\bm{y})$ we have, for Type I models,
\begin{eqnarray}
\mathbb{E}_{\bm{\theta},\bm{Y}_1,\dots,\bm{Y}_K|\bm{y}}[\hat{h}(\bm{\theta}|\bm{y})] & = & \mathbb{E}_{\bm{\theta}|\bm{y}} \left[ \mathbb{E}_{\bm{Y}_1,\dots,\bm{Y}_K|\bm{\theta}} \left[ \Delta_{\bm{\theta}} P(\bm{\theta}) + \nabla_{\bm{\theta}} P(\bm{\theta}) \cdot \hat{\bm{u}}(\bm{\theta}|\bm{y}) \right] \right] \\
& = & \mathbb{E}_{\bm{\theta}|\bm{y}} \left[ \Delta_{\bm{\theta}} P(\bm{\theta}) + \nabla_{\bm{\theta}} P(\bm{\theta}) \cdot \bm{u}(\bm{\theta}|\bm{y}) \right],
\end{eqnarray}
with the analogous result holding for Type II models. 
The remainder follows from \cite{Mira}:
Using the definition of the score $\bm{u}(\bm{\theta}|\bm{y})$ we have
\begin{eqnarray}
& = & \int_{\Theta} [\Delta_{\bm{\theta}}P(\bm{\theta})] p(\bm{\theta}|\bm{y}) + [\nabla_{\bm{\theta}}P(\bm{\theta})] \cdot [\nabla_{\bm{\theta}} p(\bm{\theta}|\bm{y})] d\bm{\theta}.
\end{eqnarray}
Then applying the divergence theorem \citep[see e.g.][]{Kendall} we obtain
\begin{eqnarray}
& = & \int_{\Theta} \nabla_{\bm{\theta}} \cdot [[\nabla_{\bm{\theta}} P(\bm{\theta})] p(\bm{\theta}|\bm{y})] d\bm{\theta} \\
& = & \oint_{\partial \Theta} [[\nabla_{\bm{\theta}} P(\bm{\theta})] p(\bm{\theta}|\bm{y})] \cdot \bm{n}(\bm{\theta}) d\bm{\theta}.
\end{eqnarray}
The assumption of the Lemma forces this integral to equal zero, as required.
\end{proof}
To illustrate the mildness of these conditions, observe that in the case of a scalar parameter $\theta \in \mathbb{R}$ and a degree-one polynomial $P$, the boundary condition is satisfied whenever $\lim_{\theta \rightarrow \pm \infty}p(\theta|\bm{y}) = 0$.
More generally, it follows from the work of \cite{Oates} that, for unbounded state spaces $\Theta \subseteq \mathbb{R}^d$, a sufficient condition for unbiasedness is that the tails of $p(\bm{\theta}|\bm{y})$ vanish faster than $\|\bm{\theta}\|^{d+k-2}$ where $k$ is the degree of the polynomial $P$.
(Here $\|\cdot\|$ can be taken to be any norm on $\mathbb{R}^d$, due to the equivalence of norms in finite dimensions.)

\subsection{Optimising the tuning parameters} \label{optimisation}

Our proposed estimator has two tuning parameters; (i) the polynomial coefficients $\bm{\phi}$, and (ii) the number $K$ of forward-simulations from $\bm{Y}|\bm{\theta}$, in the case of Type I intractability, or from $\bm{X}|\bm{\theta},\bm{y}$ in the case of Type II intractability. 
In this section we derive optimal choices for both of these tuning parameters.
Here optimality is defined as maximising the variance reduction factor, that in the case of Type I models is defined as
\begin{eqnarray}
R := \frac{\mathbb{V}_K[g(\bm{\theta})]}{\mathbb{V}_K[g(\bm{\theta}) + \hat{h}(\bm{\theta}|\bm{y})]}
\label{R}
\end{eqnarray}
where the subscript $K$ indicates that randomness arises from the augmented posterior $p(\bm{\theta},\bm{Y}_1,\dots,\bm{Y}_K|\bm{y})$.
The case of Type II models simply replaces $\bm{Y}_1,\dots,\bm{Y}_K$ with $\bm{X}_1,\dots,\bm{X}_K$.
Below we proceed by firstly deriving the optimal coefficients $\bm{\phi}^*$ for fixed number $K$ of simulations and subsequently deriving the optimal value of $K$ assuming the use of optimal coefficients.

\subsubsection{Polynomial coefficients $\bm{\phi}$} \label{phi hat}

First we consider the optimal choice of polynomial coefficients $\bm{\phi}$; this follows fairly straight-forwardly from classical results.
For general degree polynomials $P(\bm{\theta}|\bm{\phi})$ with coefficients $\bm{\phi}$ we can write $\hat{h}(\bm{\theta}|\bm{y}) = \bm{\phi}^T \bm{m}(\bm{\theta},\hat{\bm{u}})$, where in the case of degree-one polynomials $\bm{m}(\bm{\theta},\hat{\bm{u}}) = \hat{\bm{u}}$ and for higher polynomials the map $\bm{m}$ is more complicated:
Suppose that we employ a polynomial
\begin{eqnarray}
P(\bm{\theta}) = \sum_{i=1}^d a_i \theta_i + \sum_{i,j=1}^d b_{i,j}\theta_i\theta_j + \sum_{i,j,k=1}^d c_{i,j,k}\theta_i\theta_j\theta_k + \dots
\end{eqnarray}
with coefficients $\bm{\phi} = \{a_i,b_{i,j},c_{i,j,k}, \dots\}$.
For convenience, we assume symmetries $b_{\tau(i,j)} = b_{i,j}$, $c_{\tau(i,j,k)} = c_{i,j,k}$, etc. for all permutations $\tau$.
Then from Eqn. \ref{h}
\begin{eqnarray}
\hat{h}(\bm{\theta}|\bm{y}) & = & \left[ 2\sum_{i=1}^d b_{i,i} + 6\sum_{i,j=1}^d c_{i,i,j} \theta_j + \dots \right] \nonumber \\
& & +  \sum_{i=1}^d \left[ a_i + 2\sum_{j=1}^d b_{i,j}\theta_j + 3\sum_{j,k=1}^d c_{i,j,k}\theta_j\theta_k + \dots \right] \hat{u}_i(\bm{\theta}|\bm{y}).
\end{eqnarray}
This can in turn be re-written as $\hat{h}(\bm{\theta}|\bm{y}) = \bm{\phi}^T \bm{m}(\bm{\theta},\hat{\bm{u}})$ where the components $\{a_i,b_{i,j},c_{i,j,k},\dots\}$ of $\bm{\phi}$ and $\bm{m}(\bm{\theta},\hat{\bm{u}})$ are identified in the inner product as
\begin{eqnarray}
a_i & \leftrightarrow & \hat{u}_i \\
b_{i,i} & \leftrightarrow & 2 + 2\theta_i\hat{u}_i \\
b_{i,j} & \leftrightarrow & 2\theta_j\hat{u}_i + 2\theta_i\hat{u}_j \; \; \; (i<j) \\
c_{i,i,i} & \leftrightarrow & 6\theta_i + 3\theta_i^2\hat{u}_i \\
c_{i,i,j} & \leftrightarrow & 12\theta_j + 12\theta_i\theta_j\hat{u}_i + 6 \theta_i^2\hat{u}_k \; \; \; (i<j) \\
c_{i,j,k} & \leftrightarrow & 6\theta_j\theta_k\hat{u}_i + 6\theta_i\theta_k\hat{u}_j + 6\theta_i\theta_j\hat{u}_k \; \; \; (i<j<k) \; \; \; \dots
\end{eqnarray}
An optimal choice of coefficients for general degree polynomials is given by the following:

\begin{lemma}
For Type I models, the variance reduction factor $R$ is maximised over all possible coefficients $\bm{\phi}$ by the choice
\begin{eqnarray}
\bm{\phi}^*(\bm{y}) := -\mathbb{V}_K^{-1}[\bm{m}(\boldsymbol{\theta},\hat{\bm{u}})] \mathbb{E}_K[g(\bm{\theta}) \bm{m}(\boldsymbol{\theta},\hat{\bm{u}})] \label{optimal phi} 
\end{eqnarray}
and, at the optimal value $\bm{\phi} = \bm{\phi}^*$, we have
\begin{eqnarray}
R^{-1} = 1 - \rho(K)^2 
\label{corr min}
\end{eqnarray}
where $\rho(K) = \text{\emph{Corr}}_K[g(\bm{\theta}),\hat{h}(\bm{\theta}|\bm{y})]$.
An analogous result holds for Type II models, replacing $\bm{Y}_1,\dots,\bm{Y}_K$ with $\bm{X}_1,\dots,\bm{X}_K$.
\end{lemma}
\begin{proof}
This is a standard result in control variate theory for a linear combination of (well-defined) control variates \citep[e.g. p. 664,][]{Rubinstein}.
\end{proof}

Following the recommendations of \cite{Mira,Papamarkou,Oates} we mainly restrict attention to polynomials of degree at most two.
Indeed, degree-two polynomials are sufficient for exactness in the special cases discussed in \cite{Papamarkou}.
Similarly following \cite{Mira}, we estimate $\bm{\phi}^*$ by plugging in the empirical variance and covariance matrices into Eqn. \ref{optimal phi} to obtain an estimate $\hat{\bm{\phi}}$.
This introduces estimator bias since the same samples
are ``used twice'', however \cite{Glasserman} argues that this bias vanishes more quickly than the Monte Carlo error and hence the error due to this plug-in procedure is typically ignored.
(Any bias could alternatively be removed via a sample-splitting step, but this does not seem necessary for the examples that we consider below.)

\subsubsection{Number of forward-simulations $K$}

Now we derive an optimal number $K$ of forward-simulations to generate at each state $\bm{\theta}^{(i)}$ visited in the MCMC sample path, assuming the use of optimal coefficients as derived above.
Assuming that parallel computations occur no additional cost, this optimum will depend on the number $K_0$ of cores that are available for parallel processing in the computing architecture and we consider the general case below.
We present the following Lemma for Type I models, but the analogous result holds for Type II models by simply replacing $\bm{Y}_1,\dots,\bm{Y}_K$ with $\bm{X}_1,\dots,\bm{X}_K$.

\begin{lemma} \label{CLT}
Assume that (i) the condition of Lemma \ref{conditions} is satisfied, (ii) perfect transitions of the Markov chain (i.e. perfect mixing) 
is achieved
(iii) $\mathbb{E}_K[(g(\bm{\theta}) + \hat{h}(\bm{\theta}|\bm{y}))^2] <\infty$, and (iv) $\bm{\phi} = \bm{\phi}^*$. Then
\begin{eqnarray}
\sqrt{I} (\hat{\mu} - \mu) \xrightarrow{d} N\left( 0, (1-\rho(K)^2)\mathbb{V}_{\bm{\theta}|\bm{y}}[g(\bm{\theta})] \right).
\end{eqnarray}
\end{lemma}
\begin{proof}
From (i) we have that $\mathbb{E}_K[g(\bm{\theta}) + \hat{h}(\bm{\theta}|\bm{y})] = \mu$.
From (ii), (iii) and the central limit theorem we have that 
\begin{eqnarray}
\sqrt{I} ( \hat{\mu} - \mu ) \xrightarrow{d} N(0,\mathbb{V}_K[g(\bm{\theta}) + \hat{h}(\bm{\theta}|\bm{y})]).
\end{eqnarray}
Then from (iv) and Eqn. \ref{R} we have that 
\begin{eqnarray}
\mathbb{V}_K[g(\bm{\theta}) + \hat{h}(\bm{\theta}|\bm{y})] = (1-\rho(K)^2) \mathbb{V}_{\bm{\theta}|\bm{y}}[g(\bm{\theta})]
\end{eqnarray}
as required.
\end{proof}

Write $I$ for the number of MCMC iterations.
Then, under the hypotheses of Lemma \ref{CLT}, the key quantity that we aim to minimise is the cost-normalised variance ratio
\begin{eqnarray}
r(K,I) := \frac{1-\rho(K)^2}{I}, \label{tradeoff}
\end{eqnarray}
where the optimisation is constrained by fixed computational cost $c = I \ceil{K/K_0}$ on a $K_0$-core architecture.
In other words, for fixed computational cost $c$, should we focus on obtaining more MCMC samples (large $I$) or better estimating the 
RV 
control variates (large $K$)?
(Note that we assume the calculation of the score vector incurs negligible computational cost - this is certainly true whenever the score is itself a pre-requisite for MCMC sampling.)
This resource-allocation problem can be solved analytically:

\begin{lemma} \label{one}
The optimum variance for fixed computational cost (i.e. $c = I \ceil{K/K_0}$) is always achieved by 
setting $K=K_0$, the available number of cores.
\end{lemma}
\begin{proof}
See the Appendix.
\end{proof}

Our findings may be concisely summarised as follows:
For serial computation, choose $K=1$ and $I$ as large as possible. 
This typically requires no additional computation relative to standard estimation since one forward-simulation $\bm{Y}$ is generated as part of the exchange algorithm and at least one forward-simulation $\bm{X}$ is used as the basis for the pseudo-marginal algorithm.
For parallel computation, choose $K = K_0$ equal to the number of available cores (but no more) and then 
let $I$ be as large as possible.

Finally we note that RV control variates extend easily to the case where multiple expectations $\mu_j = \mathbb{E}_{\bm{\theta}|\bm{y}}[g_j(\bm{\theta})]$ are 
of interest.
Indeed the same MCMC output can be used to construct control variates $h^{(j)}(\bm{\theta}|\bm{y})$ specific to problem $j$ simply by re-estimating the optimal coefficients 
\begin{eqnarray}
\bm{\phi}^{*,(j)}(\bm{y}) = - \mathbb{V}_K^{-1}[\bm{m}(\bm{\theta},\hat{\bm{u}})] \mathbb{E}_K[g_j(\bm{\theta})\bm{m}(\bm{\theta},\hat{\bm{u}})]
\end{eqnarray}
based on the target function $g_j$ and proceeding as above.
In this way multiple expectations can be estimated without requiring any additional sampling or simulation.

\section{Applications}

Here we provide empirical results for an analytically tractable example, along with a version of the Ising model 
(Type I intractability), 
an exponential random graph model (Type I) and a nonlinear stochastic differential equation model (Type II).

\subsection{Example 1: Tractable exponential}

As a simple and analytically tractable example, consider inference for the posterior mean $\mu = \mathbb{E}_{\theta|y}[\theta]$, so that $g(\theta)=\theta$, where data $y$ arise from the exponential distribution $p(y|\theta) = \theta \exp(-\theta y)$ and inference is performed using an improper prior $p(\theta) \propto 1$.
The exponential likelihood can be formally viewed as a GRF with sufficient statistic $s(y) = -y$ and partition function 
$\mathfrak{P}(\theta) = \frac{1}{\theta}$, 
however the model is sufficiently simple that all quantities of interest are available in closed form.
Indeed it can easily be verified that $p(\theta|y) = y^2\theta\exp(-\theta y)$, so that the posterior is directly seen to satisfy the boundary condition of Lemma \ref{conditions} for any polynomial.
The true posterior expected value is $\mu = \frac{2}{y}$ and similarly the score function can be computed exactly as $u(\theta|y) = -y + \frac{1}{\theta}$.

\begin{figure}[t!]
\includegraphics[width = 0.49\textwidth]{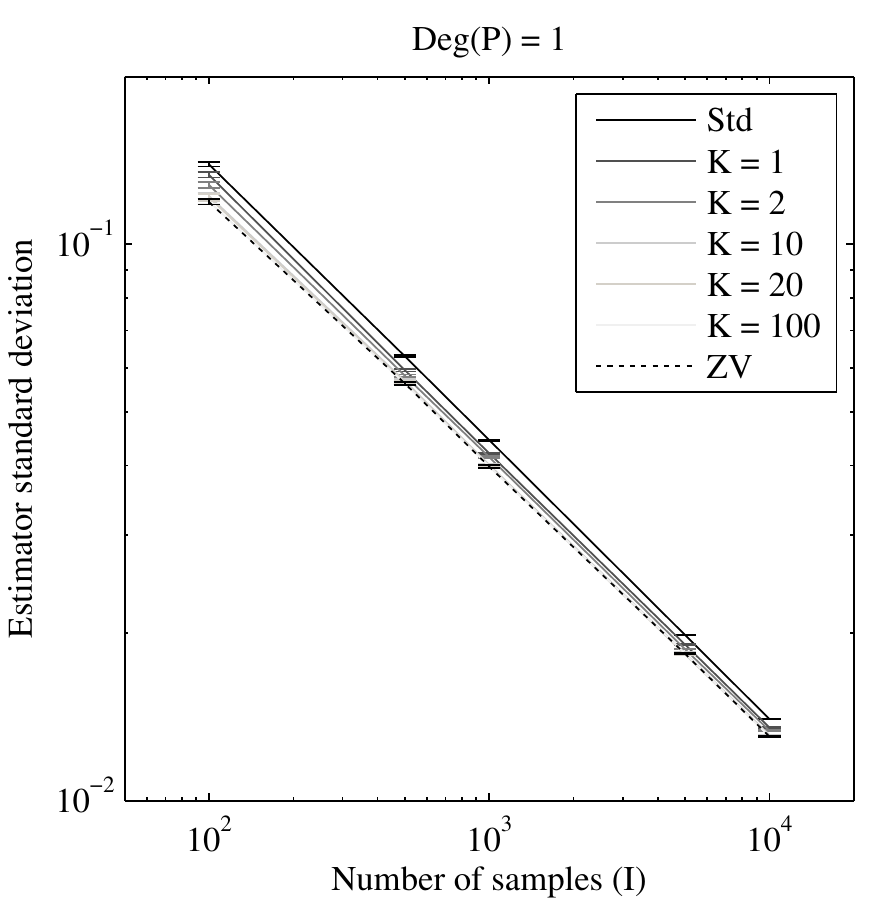}
\includegraphics[width = 0.49\textwidth]{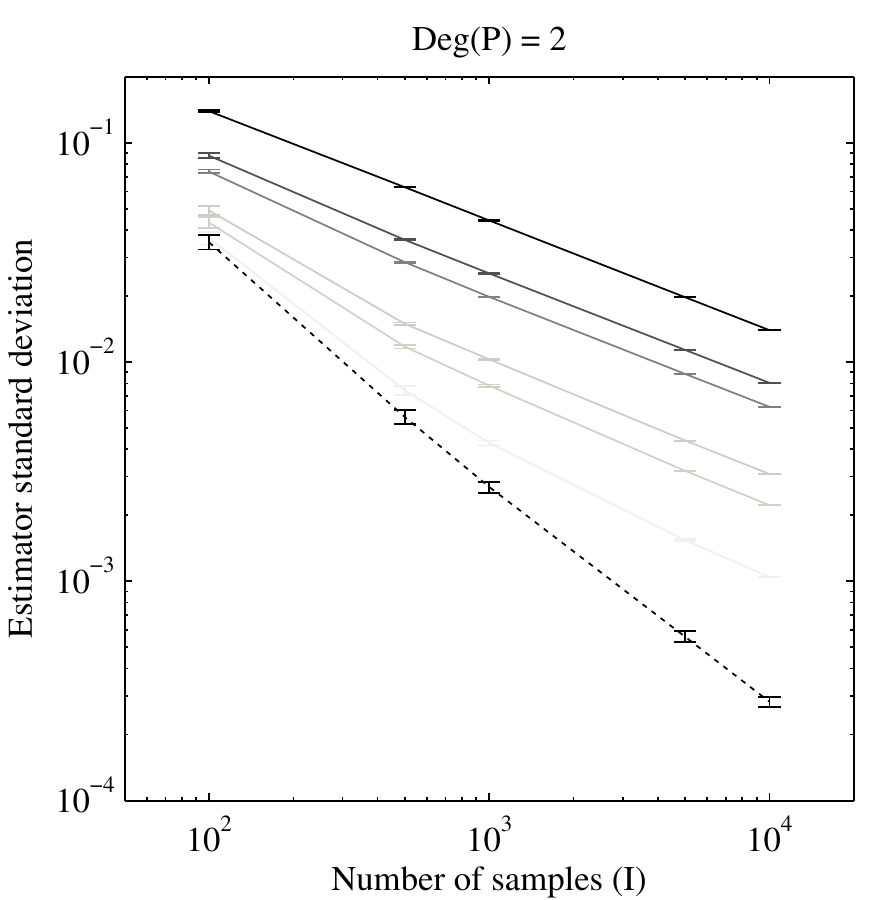}\\
\includegraphics[width = 0.49\textwidth]{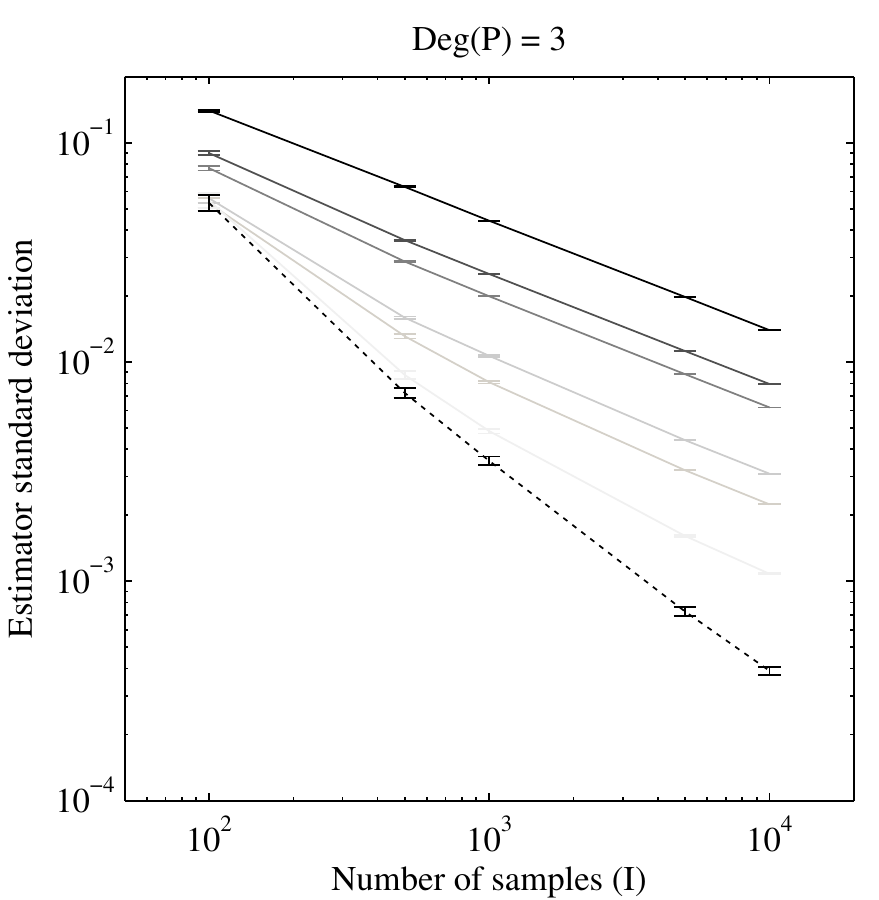}
\includegraphics[width = 0.49\textwidth]{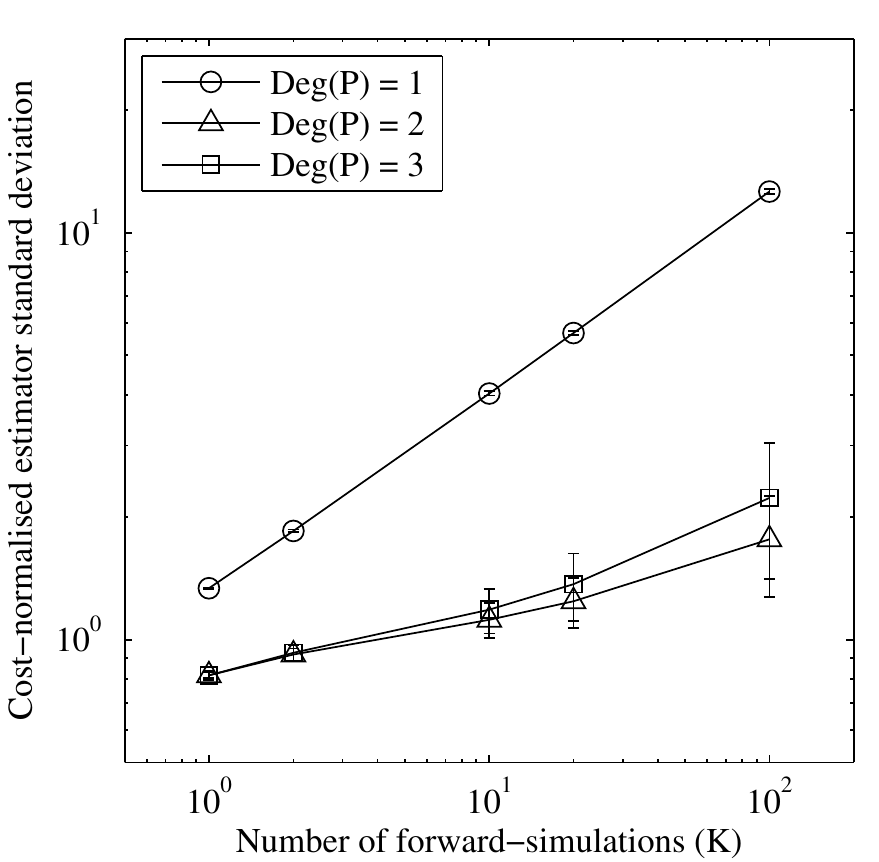}
\caption{
Tractable exponential example.
Comparing the standard deviations of Monte Carlo estimators, including the default estimator (``Std'';  the regular MCMC estimator with no variance reduction), 
the reduced-variance estimator with $K=1$, 2, 10, 20 and 100, and the ZV estimator.
[The ZV estimate has some non-zero standard deviation here because, in practice, the optimal coefficients $\bm{\phi}(y)^*$ must be estimated using Monte Carlo.]
The final panel (bottom right) displays estimator standard deviation normalised by computational cost.}
\label{tablefigure}
\end{figure}

All of the estimators that we consider are (essentially) unbiased (as noted before, the negligible bias resulting from estimation of $\hat{\bm{\phi}}$ can trivially be removed by data-splitting); in this section we therefore restrict attention to examining the estimator variances.
The maximum variance reduction that we achieve with access to the exact score 
can be obtained from $\rho(\infty) = \text{Corr}_{\theta|y}(\theta,\bm{\phi}^*(y)^T\bm{m}(\theta,u))$.
For degree-one polynomials $P(\theta) = a\theta$ the ZV method corresponds to $m(\theta,u) = u = -y + \frac{1}{\theta}$ and, since $\theta$ is not strongly linearly correlated with $\frac{1}{\theta}$, the maximum variance reduction that can be achieved by degree-one polynomials is not substantial.
However the use of degree-two polynomials $P(\theta) = a\theta + b\theta^2$ leads to $\bm{m}(\theta,u) = [u , 2+2\theta u] = [-y + \frac{1}{\theta} , 4 - 2y \theta]$ and taking $\bm{\phi} = [0,\frac{1}{2y}]$ leads to a control variate $\bm{\phi}^T\bm{m}(\theta,u) = \frac{2}{y} - \theta$.
Thus the ZV estimator $g(\theta) + \bm{\phi}^T\bm{m}(\theta,u)$ is equal to $\frac{2}{y}$, which is independent of $\theta$, i.e. exact zero variance is achieved.

In general the score $u(\theta|y)$ will be unavailable for GRF but may be estimated by $\hat{u}(\theta|y)$ as described above, with the estimate becoming exact as $K \rightarrow \infty$.
We investigate through simulation the effect of employing finite values of $K$.
Intuitively the proposed approach will be more effective when the target function $g(\theta)$ of interest is strongly correlated (under the posterior) with a linear combination $\bm{\phi}^T\bm{m}(\theta,\hat{u})$.
Fig. S1 demonstrates that when $K$ is large, the RV
control variates (for degree-two polynomials) are closely correlated with the ZV control variates (left column) and, hence, with the target function $g(\theta)$ (right column).
We would therefore expect to see a large reduction in Monte Carlo variance using RV
estimation in this regime.

The main conclusions to be drawn from this tractable example are summarised in Fig. \ref{tablefigure}, where we display estimates for the estimator standard deviation $\text{std}[\hat{\mu}]$, computed as the standard error of the mean over all $I$ Monte Carlo samples.
In total the estimation procedure was repeated 100 times and we report the mean value of $\text{std}[\hat{\mu}]$ along with the standard error of this mean computed over the 100 realisations.
We considered varying the number of Monte Carlo samples $I = 100,500,1000,5000,10000$, the number of forward-simulations $K = 1,2,10,20,100$ and the degree of the polynomial trial function $\text{deg}(P) = 1,2,3$.
Results demonstrate that estimator variance reduces as either $I$ or $K$ is increased, as expected.
A comparison {\it between} the plots (full data provided in  Table S1) shows that degree-two polynomials considerably out-perform the degree-one polynomials, whereas the degree-three polynomials tend to slightly under-perform the degree-two polynomials.
(The theoretical best ZV control variates are degree-two polynomials and therefore degree-three polynomials require that additional coefficients
associated to higher order control variates - that we know, theoretically, should be equal to zero - are estimated from data, thus adding extra noise.)

To assess computational efficiency, we also report the quantity $\sqrt{IK}\text{std}[\hat{\mu}]$ that has units ``standard deviation per unit serial computational cost'' and can be used to evaluate the computational efficiency of competing strategies (bottom right panel of Fig. \ref{tablefigure} and  Table S1).
Here we see that $K=1$ minimises $\sqrt{IK}\text{std}[\hat{\mu}]$ and is consistent with the theoretical result that $K=1$ is optimal for serial computation.
Fig. S2 plots the canonical correlation coefficient between $\bm{m}(\theta,\hat{u})$ and $g(\theta)$ for values of $K = 1,2,\dots,10$.
Here we notice that over $80\%$ of the correlation is captured by just one forward-simulation from the likelihood ($K=1$), further supporting our theoretical result that $K=1$ is optimal for serial computation.
Indeed, a theoretical prediction $\rho(K) = (K/(K+C))^{1/2}$ resulting from Lemma \ref{technical} in the Appendix, shown as a solid line
in Fig. S2, closely matches these simulation results.

\subsection{Example 2: Ising model}

In the experiments below we consider an Ising model of size $n = 16$, about the limit for exact solution, as defined in Sec. \ref{intro}, 
Example 1,
above.
Assuming that the lattice points have been indexed from top to bottom in each column and that columns are ordered from left to right, then an interior point  $y_i$ in a first order neighbourhood model has neighbours $\{y_{i-n}, y_{i-1}, y_{i+1},y_{i+n}\}$. 
Each point along the edges of the lattice has either two or three neighbours. 
We focus on estimating the posterior mean $\mu = \mathbb{E}_{\theta|\bm{y}}[\theta]$ under a prior $\theta \sim N(0,5^2)$. 
Since the tails of the prior vanish exponentially and the likelihood is bounded, the posterior automatically satisfies the boundary conditions of Lemma \ref{conditions}.
Here data $\bm{y}$ were simulated exactly from the likelihood using $\theta=0.4$, via the recursive scheme of \cite{Friel2}. This recursive algorithm also allows exact calculation of the partition function. 
In turn this allow a very precise estimate of the posterior mean; 
for the data that we consider below this posterior mean is $\mu = 0.43455$, calculated numerically over a very fine grid of $\theta$ values. 

\begin{figure}[h]
\centering
  \includegraphics[width=.48\textwidth]{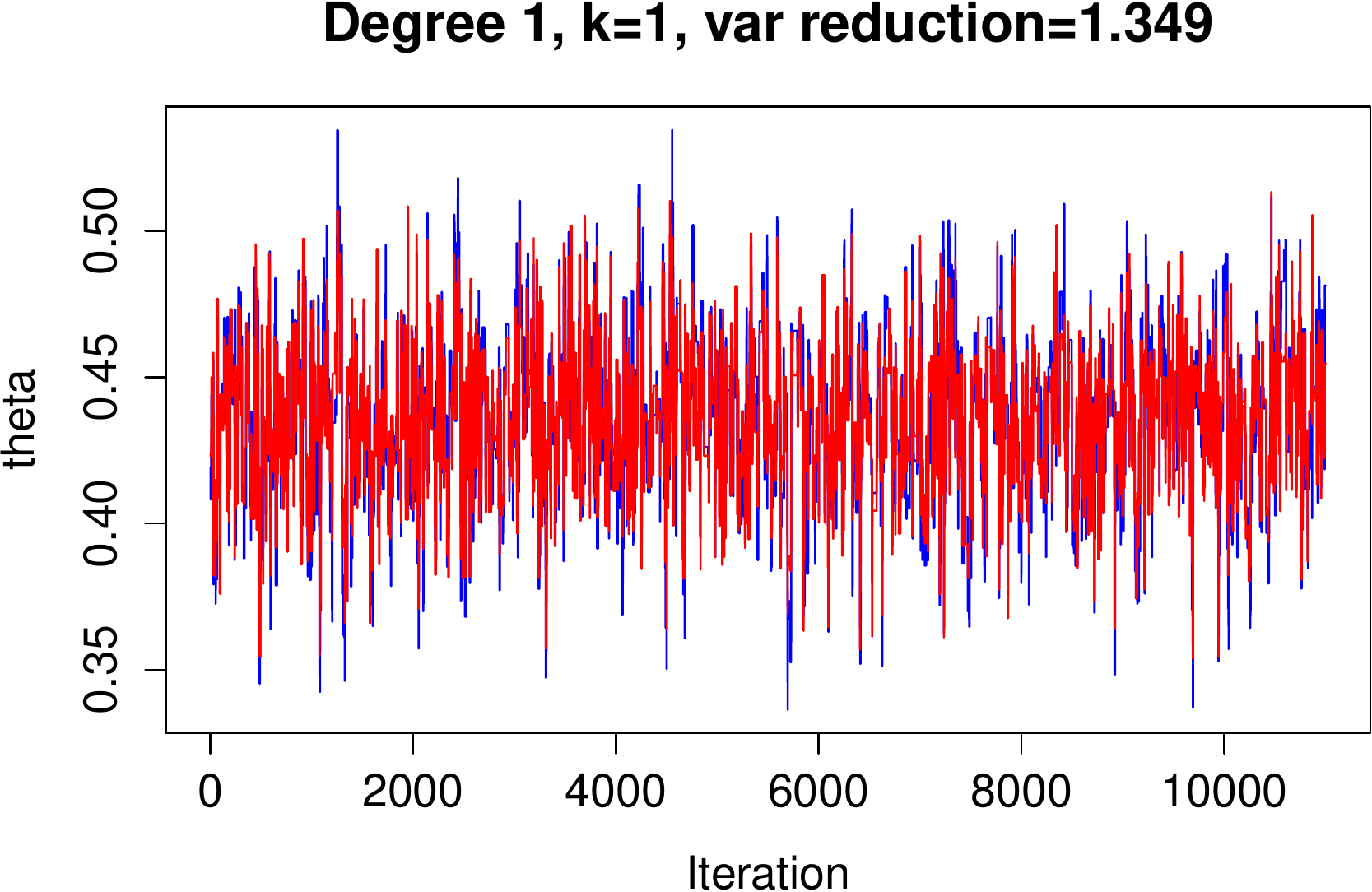} 
  \includegraphics[width=.48\textwidth]{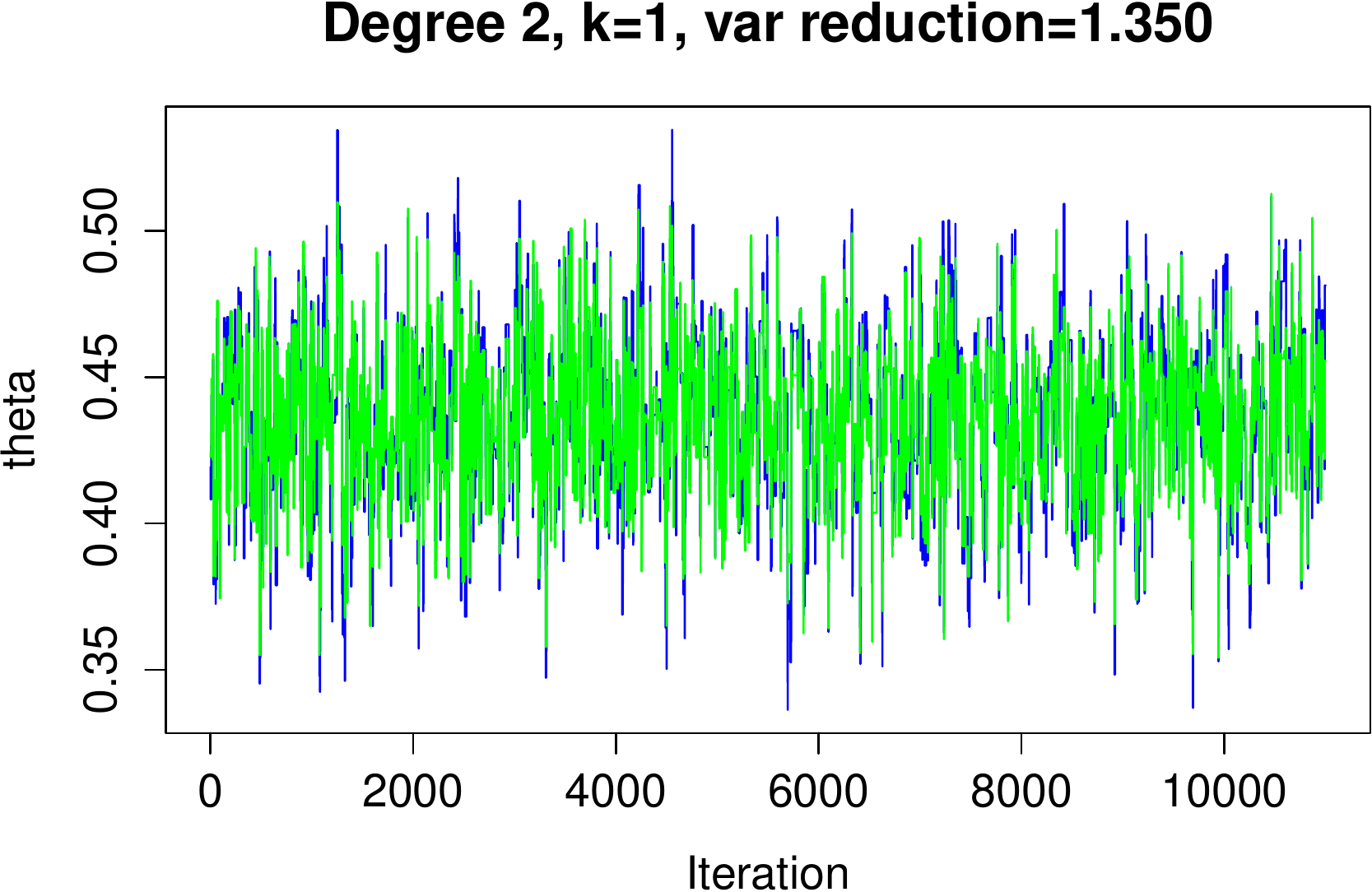} \\
  \includegraphics[width=.48\textwidth]{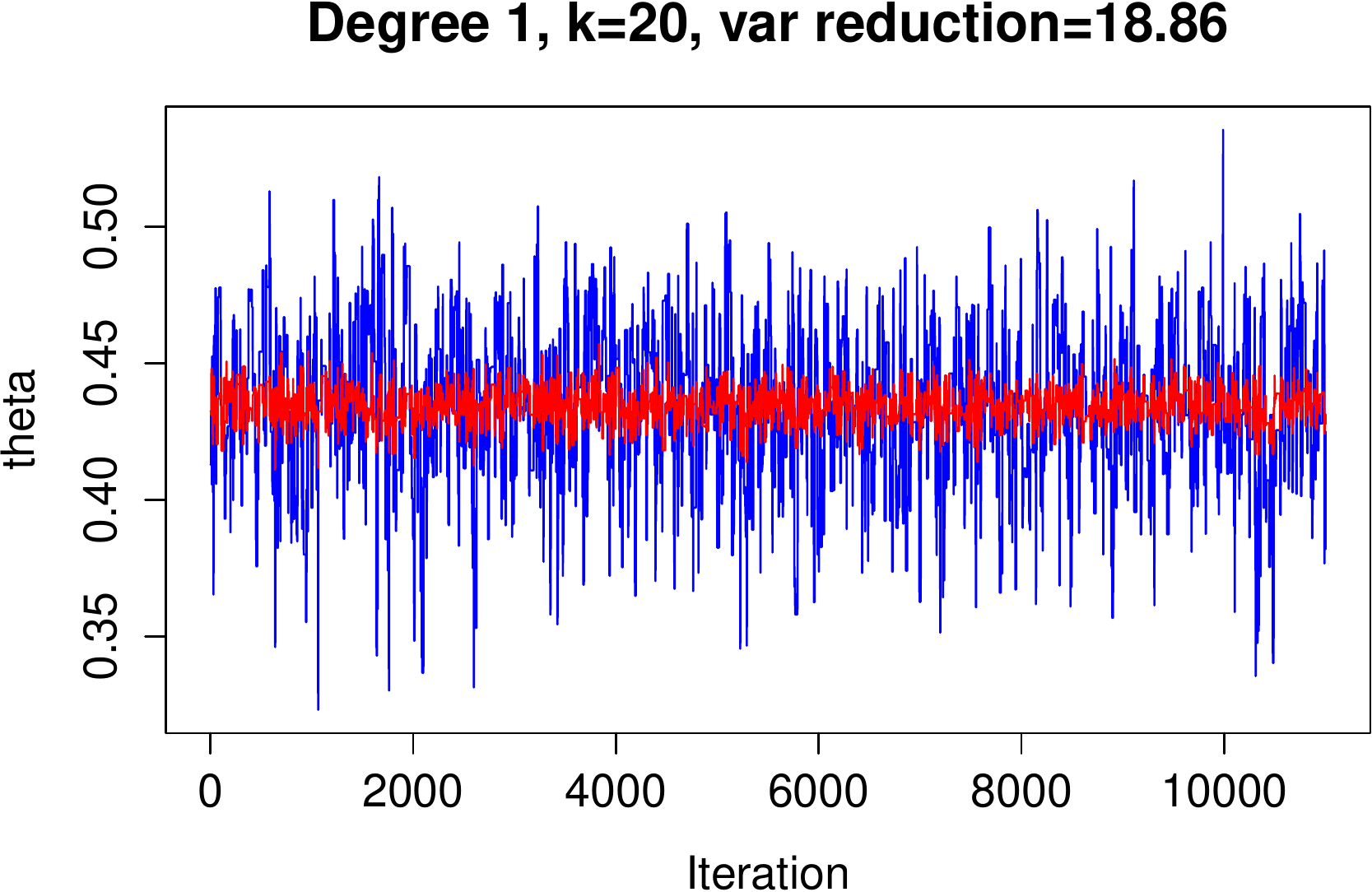}
  \includegraphics[width=.48\textwidth]{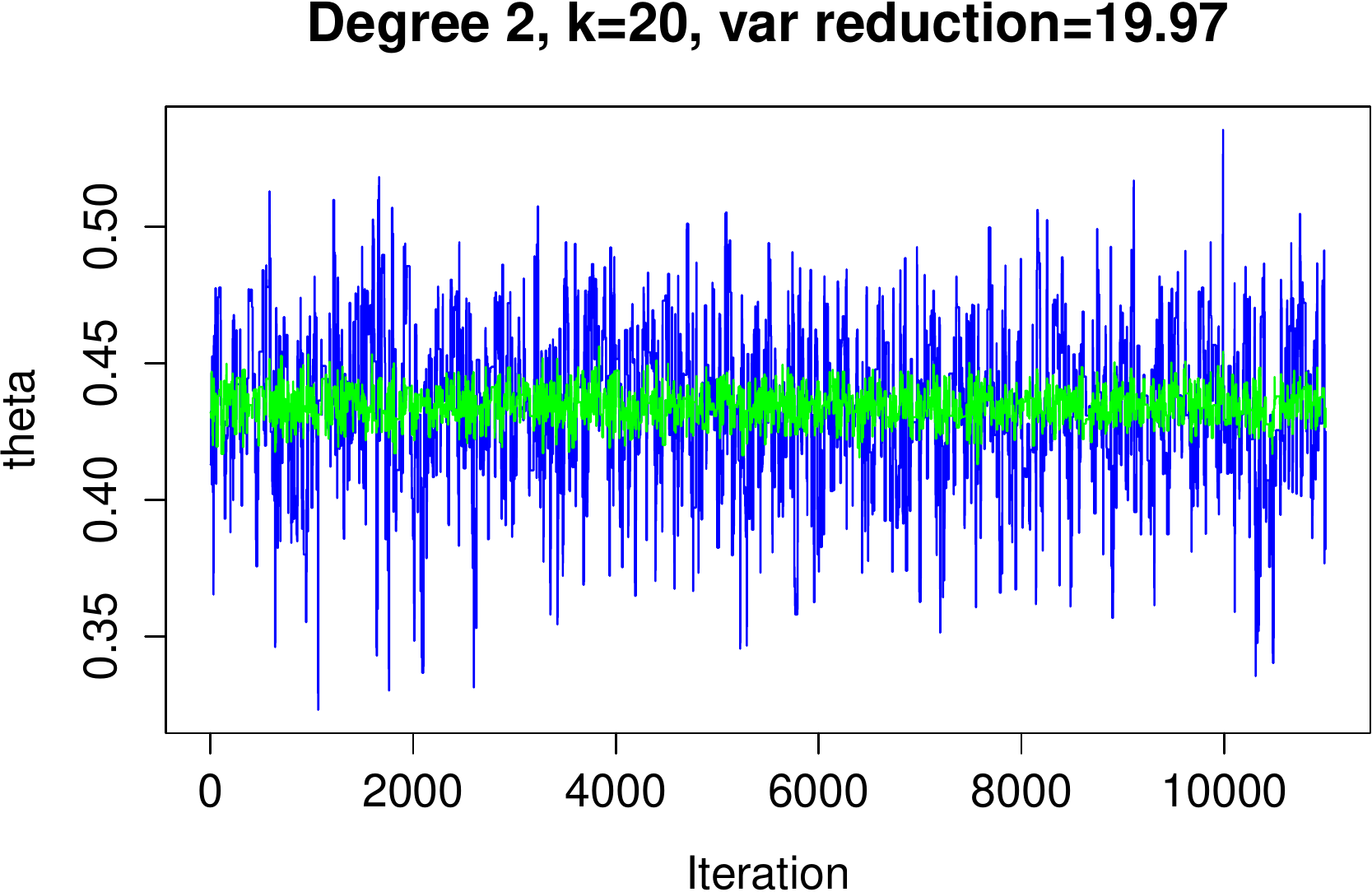} \\
  \includegraphics[width=.48\textwidth]{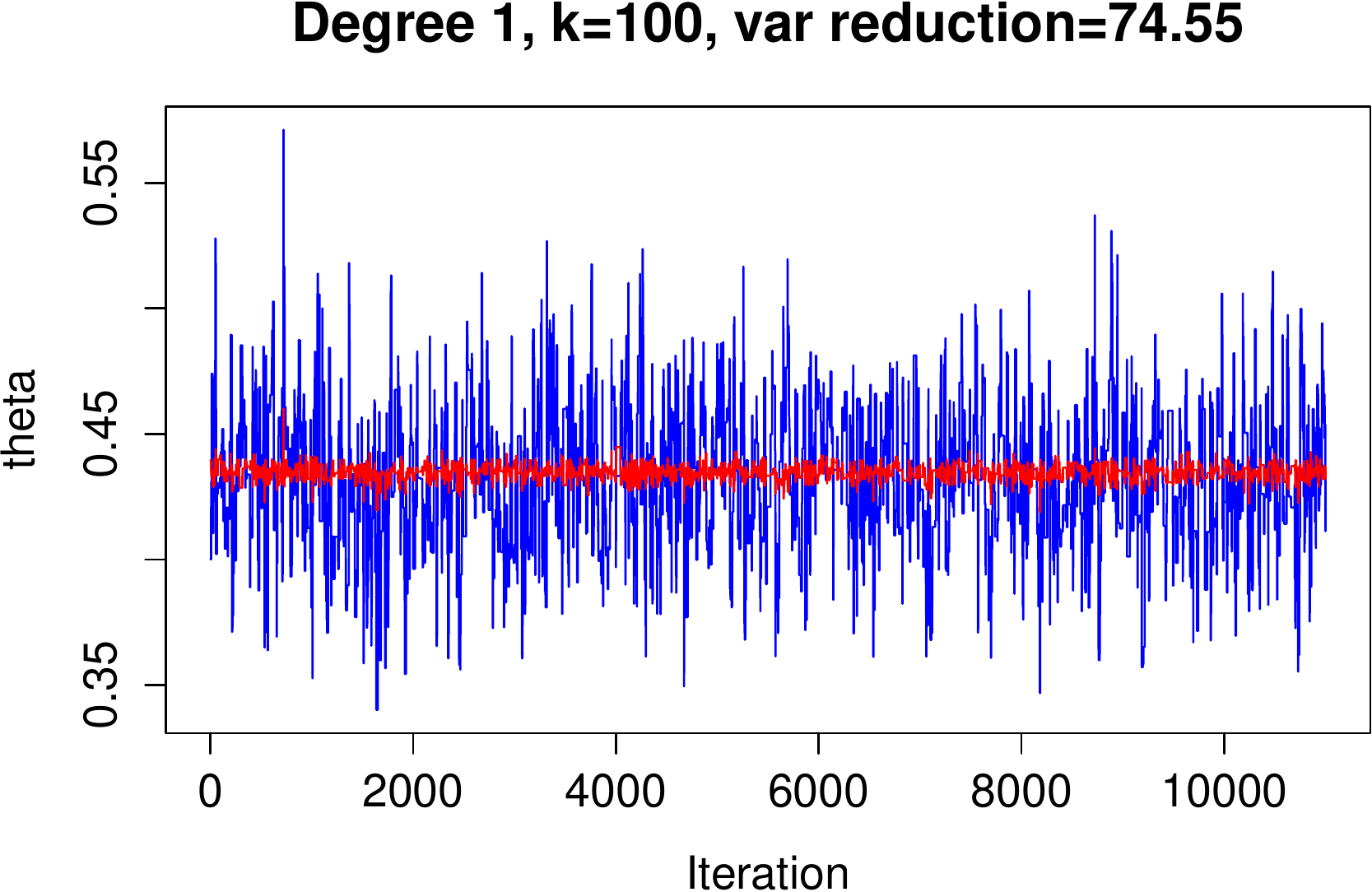} 
  \includegraphics[width=.48\textwidth]{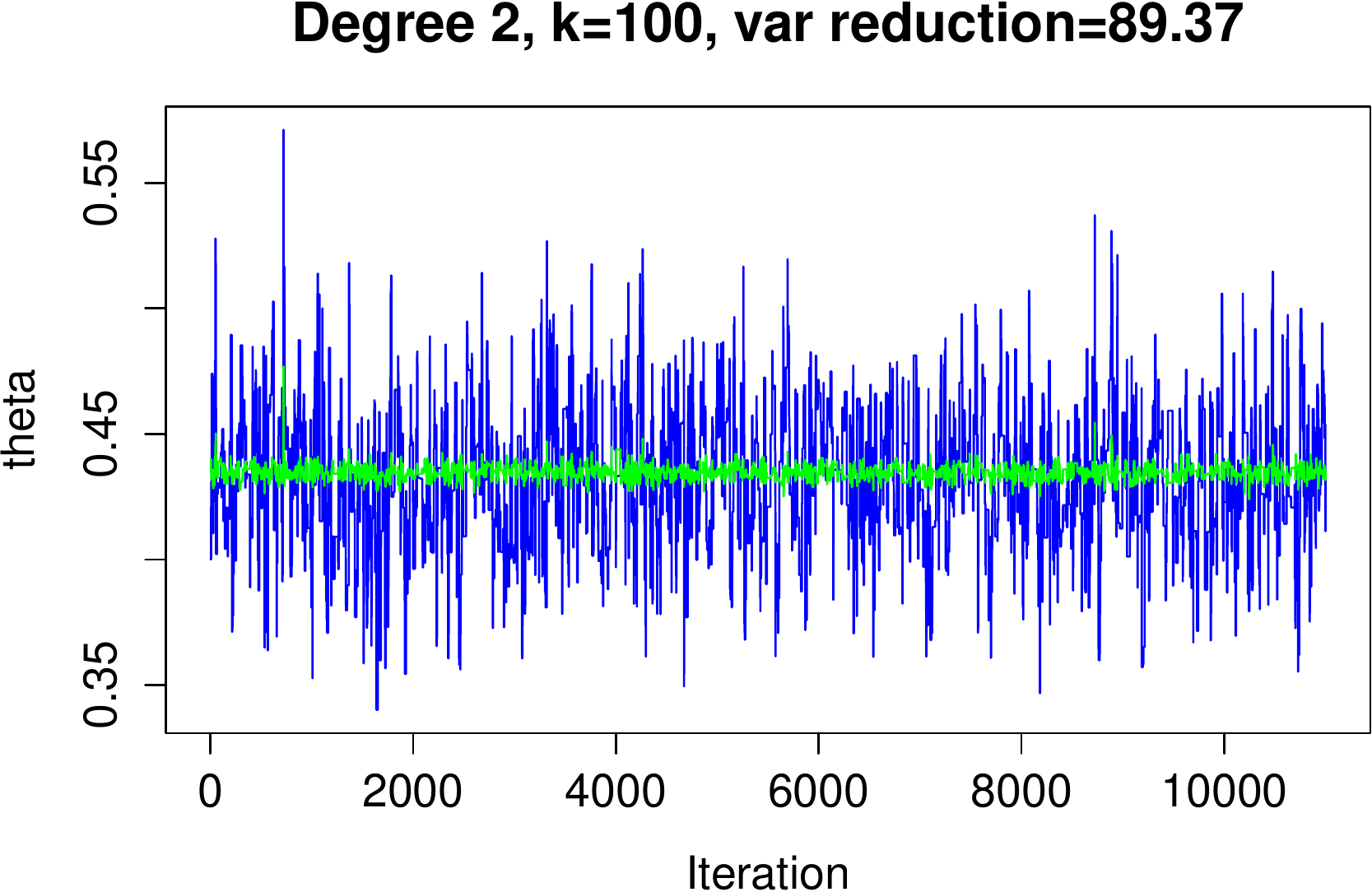} \\
  \includegraphics[width=.48\textwidth]{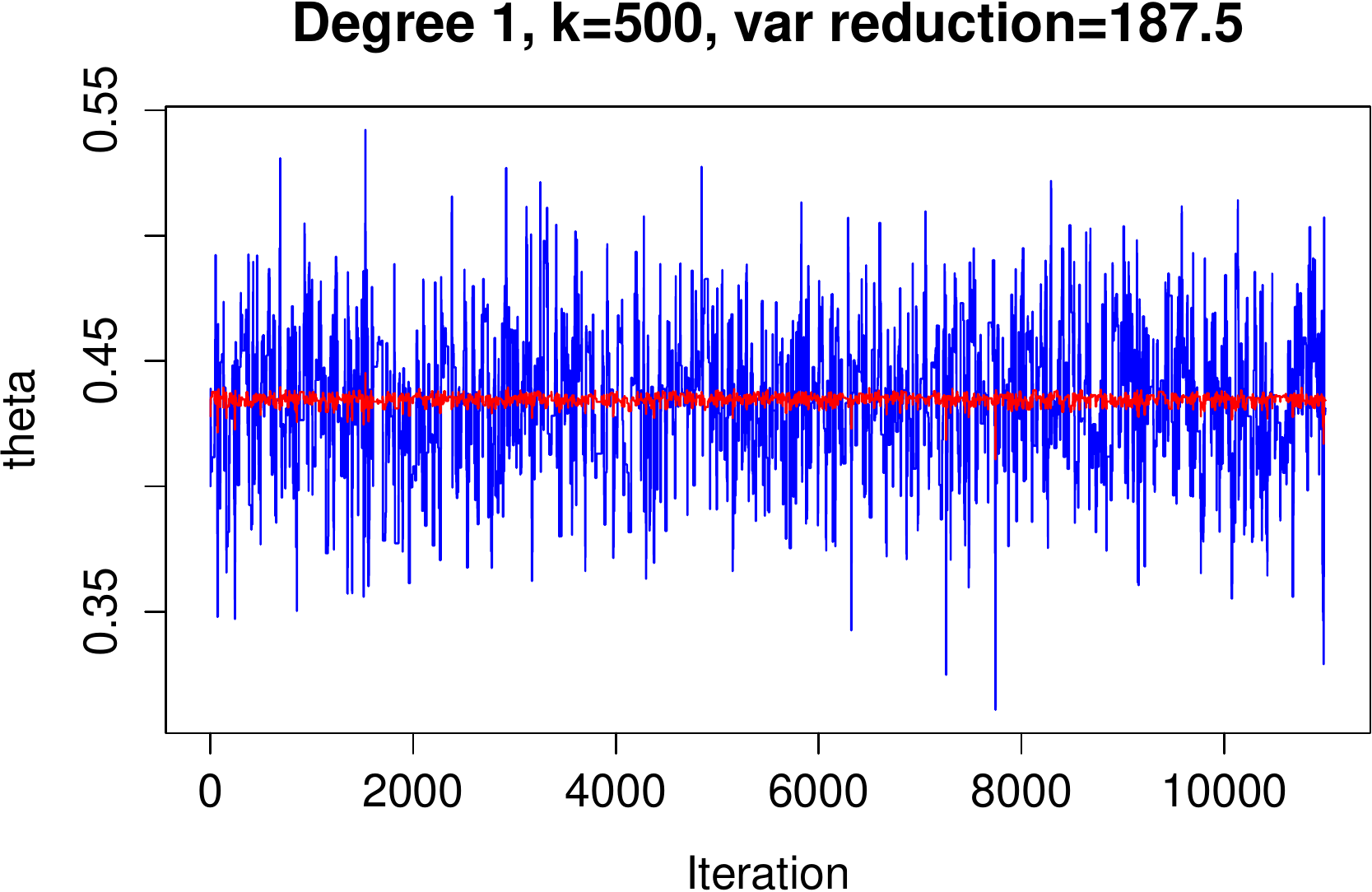} 
  \includegraphics[width=.48\textwidth]{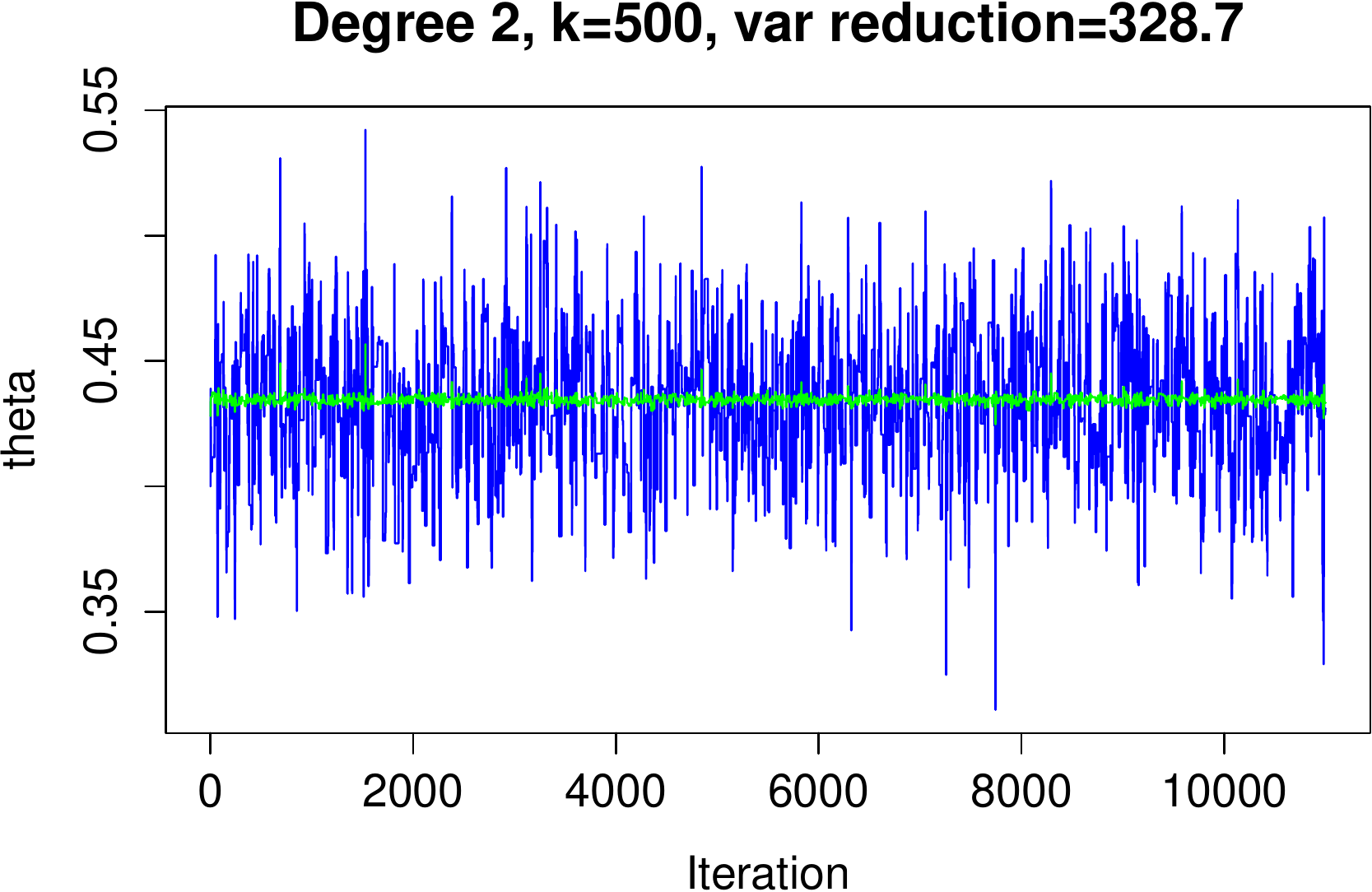} 
\caption{Ising model: As the number of forward-simulations, $K$, increases, the precision of the controlled estimate of the posterior mean for $\theta$ increases.
The degree-two polynomial yields greater precision compared to the degree-one polynomial, particularly for larger values of $K$.}
\label{fig:tempcomparison}
\end{figure}
\FloatBarrier

Fig. \ref{fig:tempcomparison} displays the MCMC trace plots, obtained using the exchange algorithm, for $g(\theta)=\theta$ (blue) and the 
RV
version $g(\theta) + \hat{\bm{\phi}}^T\bm{m}(\theta,\hat{u})$.
Trace plots are presented for increasing values of $K \in \{1,20,100,500\}$ and using degree-one (red) and degree-two (green) polynomials.\footnote{For convenience, forward-simulation was performed on a single core using a Gibbs sampler with $1,000$ burn-in iterations. 
A sample of size $K$ were collected from this chain at a lag of $500$ iterations in order to ensure that dependence between samples was negligible. 
This accurately mimics the setting of independent samples that corresponds to performing multiple forward-simulations in parallel. }
For $K=1$ we observe little difference between controlled (i.e. using RV control variates) and uncontrolled 
trajectories, suggesting that 
RV
control variates do not justify the additional coding effort in the case of serial computation. However it is evident that as $K$ increases, the Monte Carlo variance of the controlled trajectory decreases; indeed when $K = 500$ the variance is dramatically reduced compared to the (uncontrolled) MCMC samples of $\theta$.
These findings are summarised in Table \ref{tab:comparison}.
Additionally, we find that degree-two polynomials offer a substantial improvement over degree-one polynomials in terms of variance reduction, but that this is mainly realised for larger values of $K$.
These results present a powerful approach to exploit multi-core processing to deliver a real-time acceleration in the convergence of MCMC estimators.

\begin{table}[t!]
	\begin{center}
		\begin{tabular}{l l l l l}
			\hline\hline
			 \qquad & $K=1$ & $K=20$ & $K=100$ & $K=500$  \\
			\hline 
			$\hat{\mu}$ & $0.4340$ & 0.4345 & 0.4322 & 0.4340\\
                       $\hat{\mu}_1$ & 0.4351 & 0.4345 & 0.4347 & 0.4346 \\
                       $\hat{\mu}_2$ & 0.4351 & 0.4344 & 0.4346 & 0.4346\\
                       $R = \mathbb{V}[\hat{\mu}] / \mathbb{V}[\hat{\mu}_1]$  & 1.349 & 18.86 & 74.55 & 187.5\\
                       $R = \mathbb{V}[\hat{\mu}] / \mathbb{V}[\hat{\mu}_2]$ & 1.350 & 19.97 & 89.37 & 328.7 \\
			\hline
		\end{tabular}
	\end{center}
	\caption{Ising model: As the number of forward-simulations, $K$, used to estimate the score function increases, the reduction in variance becomes more substantial.
	Here $\hat{\mu}$ is the standard Monte Carlo estimate, $\hat{\mu}_1$ is the reduced-variance estimate using degree-one polynomials and $\hat{\mu}_2$ 	is the reduced-variance estimate using degree-two polynomials. 
	Brute-force calculation produces a value $\mu = 0.43455$ for this example.
	[Variances were estimated with respect to empirical means.]
	} 
	\label{tab:comparison}
\end{table}

\subsection{Example 3: Exponential random graph models}

In the experiment below we consider the \emph{Gamaneg} network \citep{Read}, displayed in Fig. \ref{fig:gamaneg}, that consists of $n=16$ sub-tribes of the Eastern central highlands of New Guinea. 
In this graph an edge represents an antagonistic relationship between two sub-tribes.
Here we consider an ERG model 
as defined in Sec. \ref{intro}, Example 2,
with 
$k=2$, i.e.
two sufficient statistics, where $s_1(y)$ counts the total number of observed edges and the two-star statistic $s_2(y)$ is also as defined in Sec. \ref{intro} above. 
Here the parameters $\theta_1, \theta_2$ control the propensity of edges and two-star configurations, respectively, in the network. 
Positive values of $\theta_1$ and $\theta_2$ tend to lead to, respectively, over-representation of
edges and two-star configurations in networks realised from the likelihood.
The prior distributions for $\theta_1$ and $\theta_2$ were both set to be independent $N(0,5^2)$, from which 
it
follows that the boundary condition of Lemma \ref{conditions} is satisfied.
This is a benchmark dataset that has previously been used to assess Monte Carlo methodology \citep{Friel}, making it well-suited to our purposes.

\begin{figure}[t!]
 \begin{center}
 \includegraphics[width=10cm]{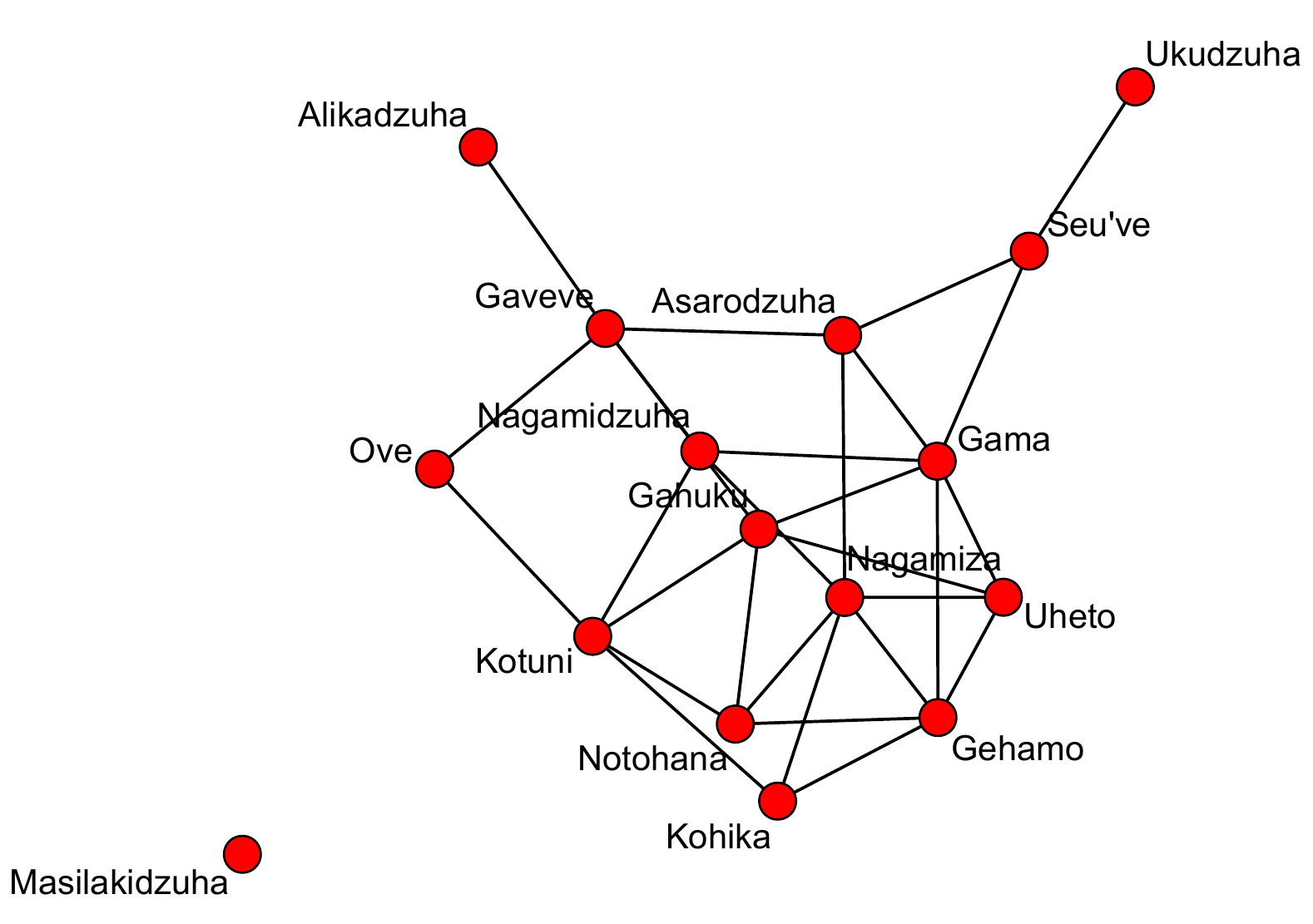}
\end{center}
\caption{Gamaneg graph. The vertices represent $16$ sub-tribes of the Eastern central highlands of New Guinea and edges represent an antagonistic relationship between two sub-tribes.}
\label{fig:gamaneg}
\end{figure}

Again we focus on the challenge of estimating the posterior mean $\bm{\mu} = \mathbb{E}_{\bm{\theta}|\bm{y}}[\bm{\theta}]$, in this case performing independent estimation with $g(\bm{\theta}) = \theta_j$ for $j=1, 2$.
Recently \cite{cai:fri11}, \cite{cai:fri14} developed Bayesian methodology for this model, based on the exchange algorithm, that can be directly utilised for the RV framework developed in this paper.  
The exchange algorithm was run for $I = 11,000$ iterations, where at each iteration $K=500$ forward-simulations were used to estimate 
the score.\footnote{For convenience, the forward-simulation step was achieved using a Gibbs sampler where a burn-in phase of $1,000$ iterations. We drew $K$ samples from this chain at a lag of $1,000$ iterations.
This accurately mimics the setting of independent samples that corresponds to performing multiple forward-simulations in parallel.} 
Fig.~\ref{fig:erg_plots} illustrates that a variance reduction of about 20 times is possible using a degree-two polynomial for each of the two components of the parameter vector.
From the uncontrolled trajectories it is difficult to comment on the relative posterior means $\mu_1$ and $\mu_2$ of $\theta_1$ and $\theta_2$ respectively, but from the controlled trajectories it is visually clear that we have $\mu_1 < \mu_2 < 0$.
This suggests that posterior predictions of network structure typically contain more two-stars than edges.

We note that \cite{Caimo} recently proposed the use of delayed rejection to reduce autocorrelation in the exchange algorithm for ERG models, demonstrating an approximate two-fold variance reduction; 
the delayed rejection exchange algorithm is fully
compatible with our methodology and, if combined, should yield a further reduction in variance.

\begin{figure}[t!]
\centering
  \includegraphics[width=.44\textwidth]{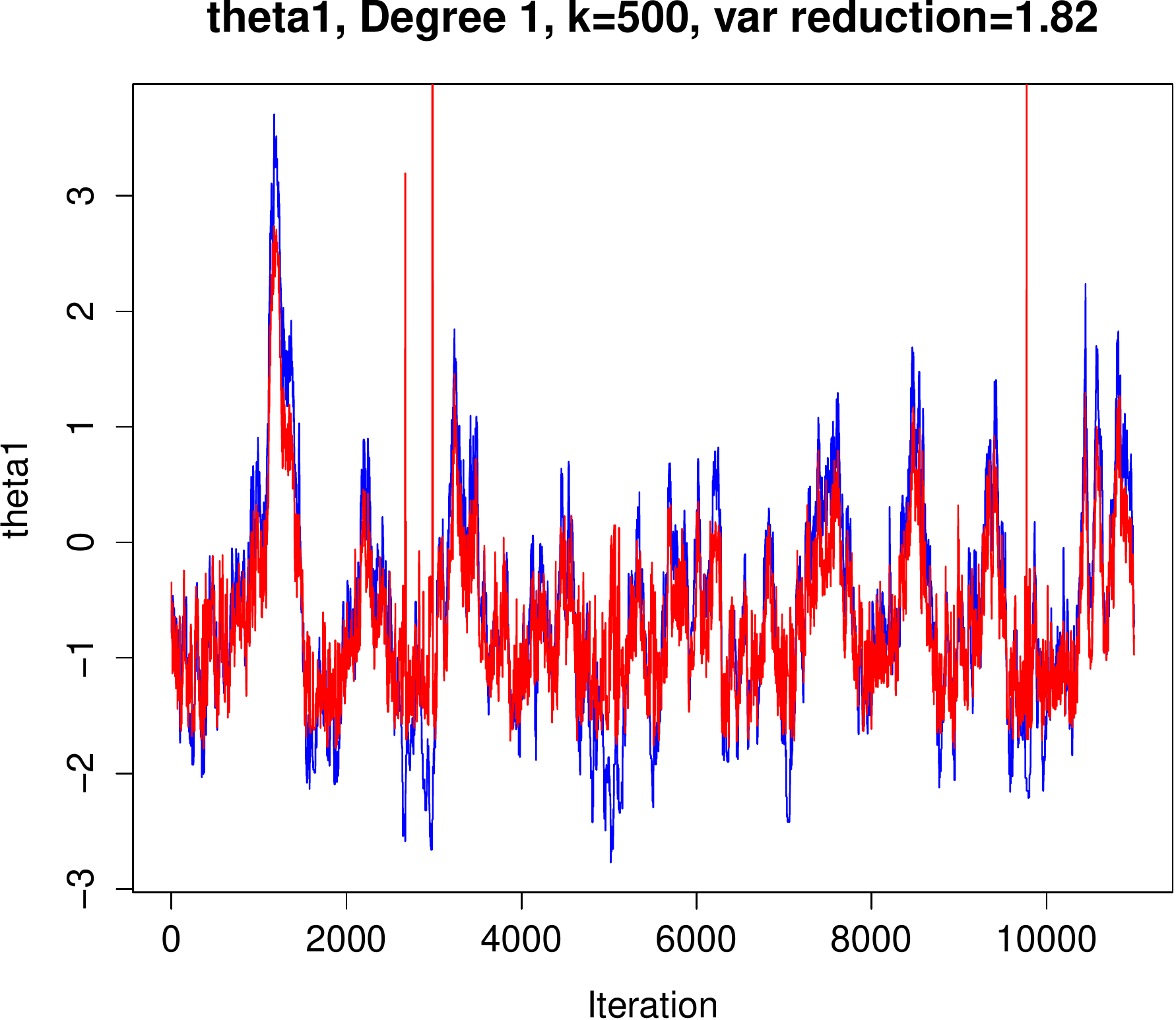} 
  \includegraphics[width=.44\textwidth]{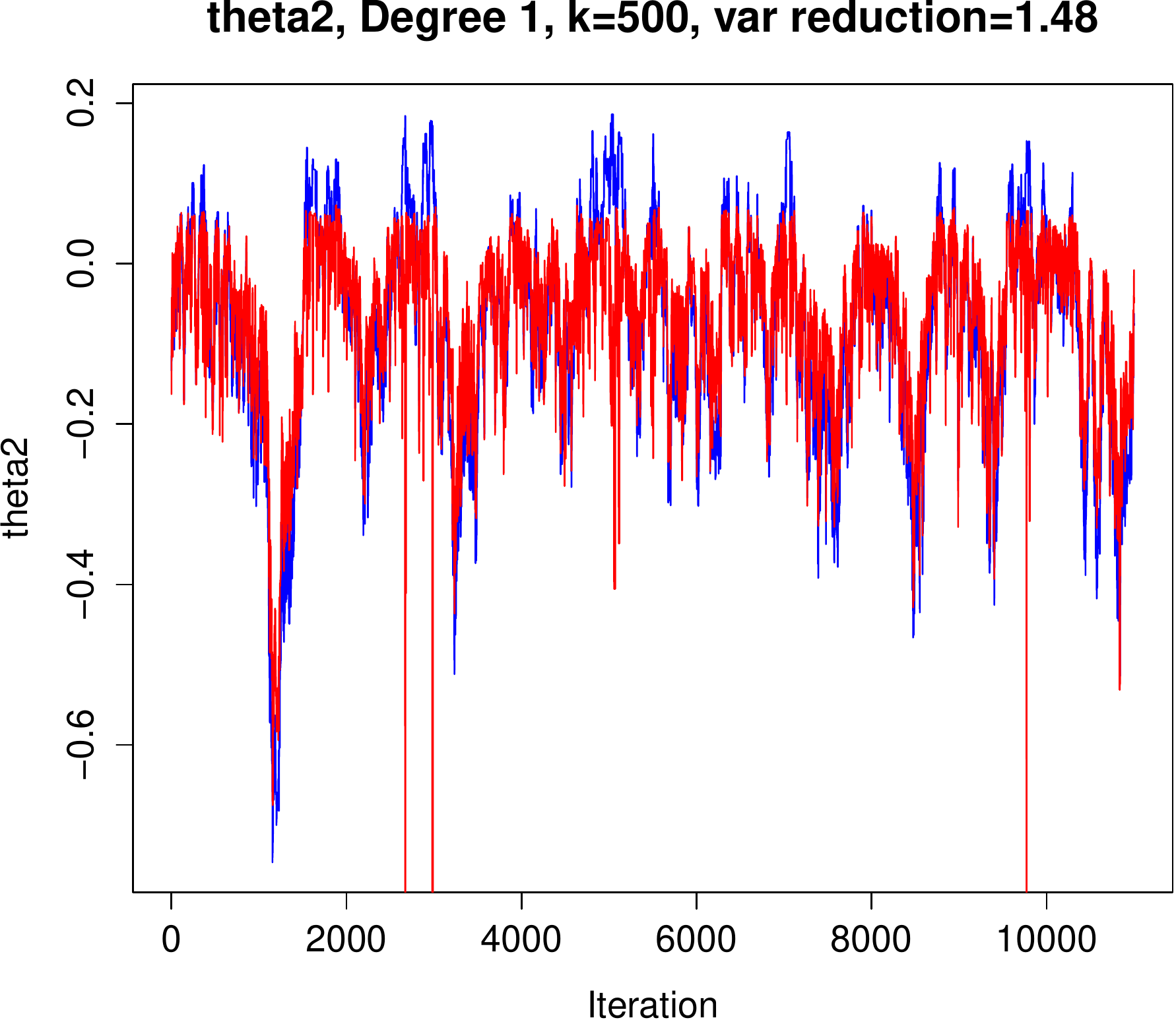} \\
  \includegraphics[width=.44\textwidth]{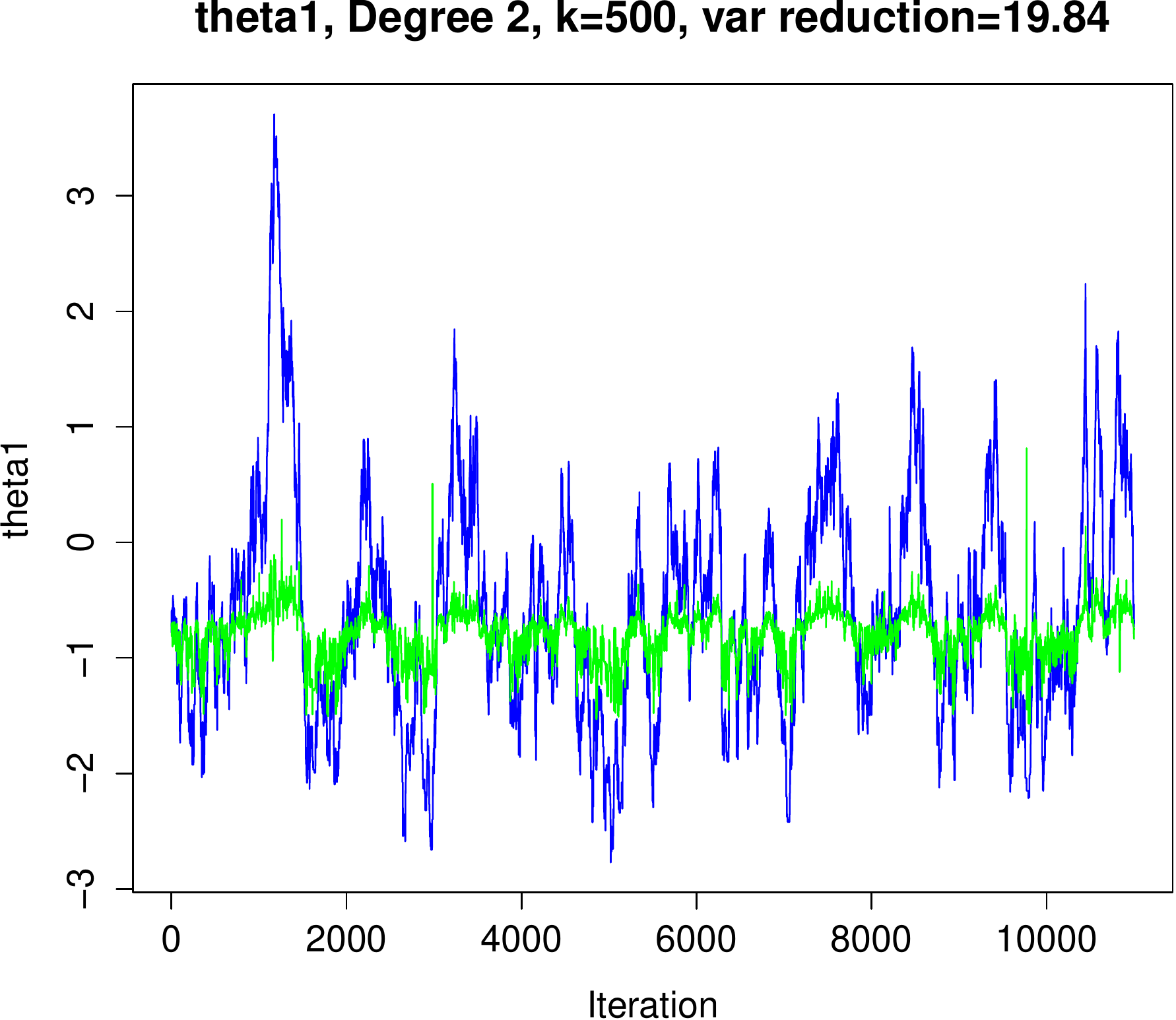} 
  \includegraphics[width=.44\textwidth]{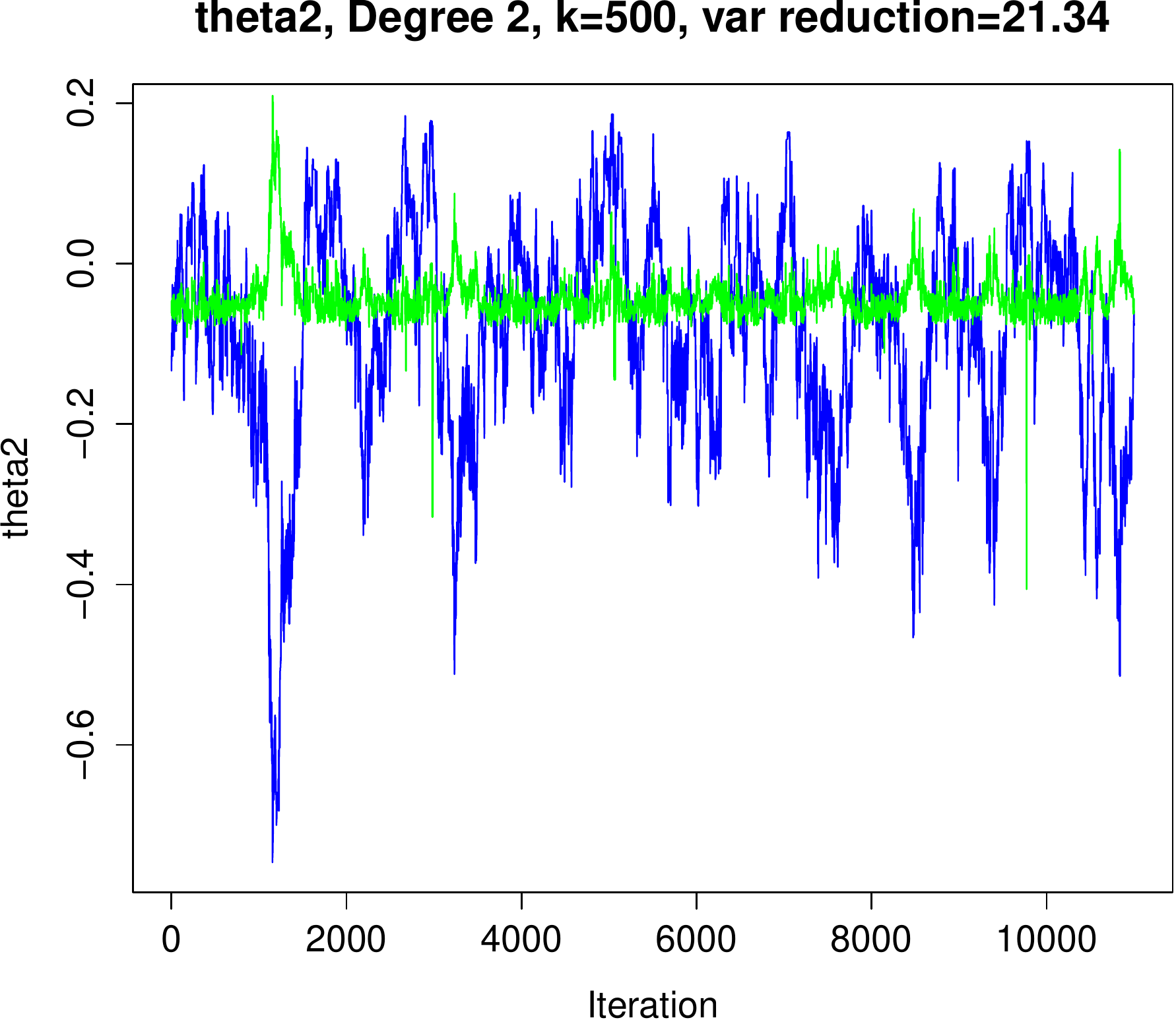} 
\caption{Exponential random graph model: The top row displays the trace plot for $\theta_1$ and $\theta_2$ in uncontrolled (blue) and controlled 
versions for a degree-one (red) polynomial, while the bottom row is similar but for a degree-two (green) polynomial.}
\label{fig:erg_plots}
\end{figure}

\FloatBarrier

\subsection{Example 4: Nonlinear stochastic differential equations}

For our final example we consider performing Bayesian inference for a system of nonlinear stochastic differential equations (SDEs)
as defined in  Sec. \ref{intro}, Example 4, Eqn. \ref{general sde}.
This problem is well-known to pose challenges for Bayesian computation and recent work in this direction includes \citep{Beskos,Golightly}.
We estimate the score using
\begin{eqnarray}
\nabla_{\bm{\theta}} \log p(\bm{\theta}|\bm{y}) \approx \frac{1}{K} \sum_{k=1}^K \nabla_{\bm{\theta}} \log p(\bm{\theta},\bm{x}^{(k)}|\bm{y}) \label{score approx}
\end{eqnarray}
where $\bm{x}^{(k)}$ are independent samples from $p(\bm{x}|\bm{\theta},\bm{y})$.
Such samples can be generated using MCMC techniques and in this paper we make use of a Metropolis-Hastings sampler with ``diffusion bridge'' proposals \citep{Fuchs}.
Note that since $ p(\bm{\theta},\bm{x}|\bm{y}) \propto  p(\bm{y},\bm{x} | \bm{\theta}) p(\bm{\theta}) = p(\bm{x}|\bm{\theta}) p(\bm{\theta})$, 
we have that 
\begin{eqnarray}
\nabla_{\bm{\theta}} \log p(\bm{\theta},\bm{x}|\bm{y}) = \nabla_{\bm{\theta}} \log p(\bm{\theta}) + \nabla_{\bm{\theta}} \log p(\bm{x} | \bm{\theta}).
\end{eqnarray}
Direct calculation shows that, assuming $\bm{\beta}$ is invertible,
\begin{eqnarray}
\nabla_{\theta_j} \log p(\bm{X} | \bm{\theta}) = \sum_{i=2}^T \begin{array}{l} -\frac{1}{2} \text{tr}(\bm{\beta}_{i}^{-1} \nabla_{\theta_j} \bm{\beta}_{i}) + (\nabla_{\theta_j}\bm{\alpha}_{i})^T \bm{\beta}_{i}^{-1} (\bm{X}_i - \bm{X}_{i-1} - \bm{\alpha}_{i} \delta t) \\
\; \; \; \; \; + \frac{1}{2\delta t} (\bm{X}_i - \bm{X}_{i-1} - \bm{\alpha}_{i} \delta t)^T \bm{\beta}_{i}^{-1} (\nabla_{\theta_j} \bm{\beta}_{i}) \bm{\beta}_{i}^{-1} (\bm{X}_i - \bm{X}_{i-1} - \bm{\alpha}_{i} \delta t) \end{array}
\end{eqnarray}

\begin{figure}[t!]
\centering
\includegraphics[width=\textwidth]{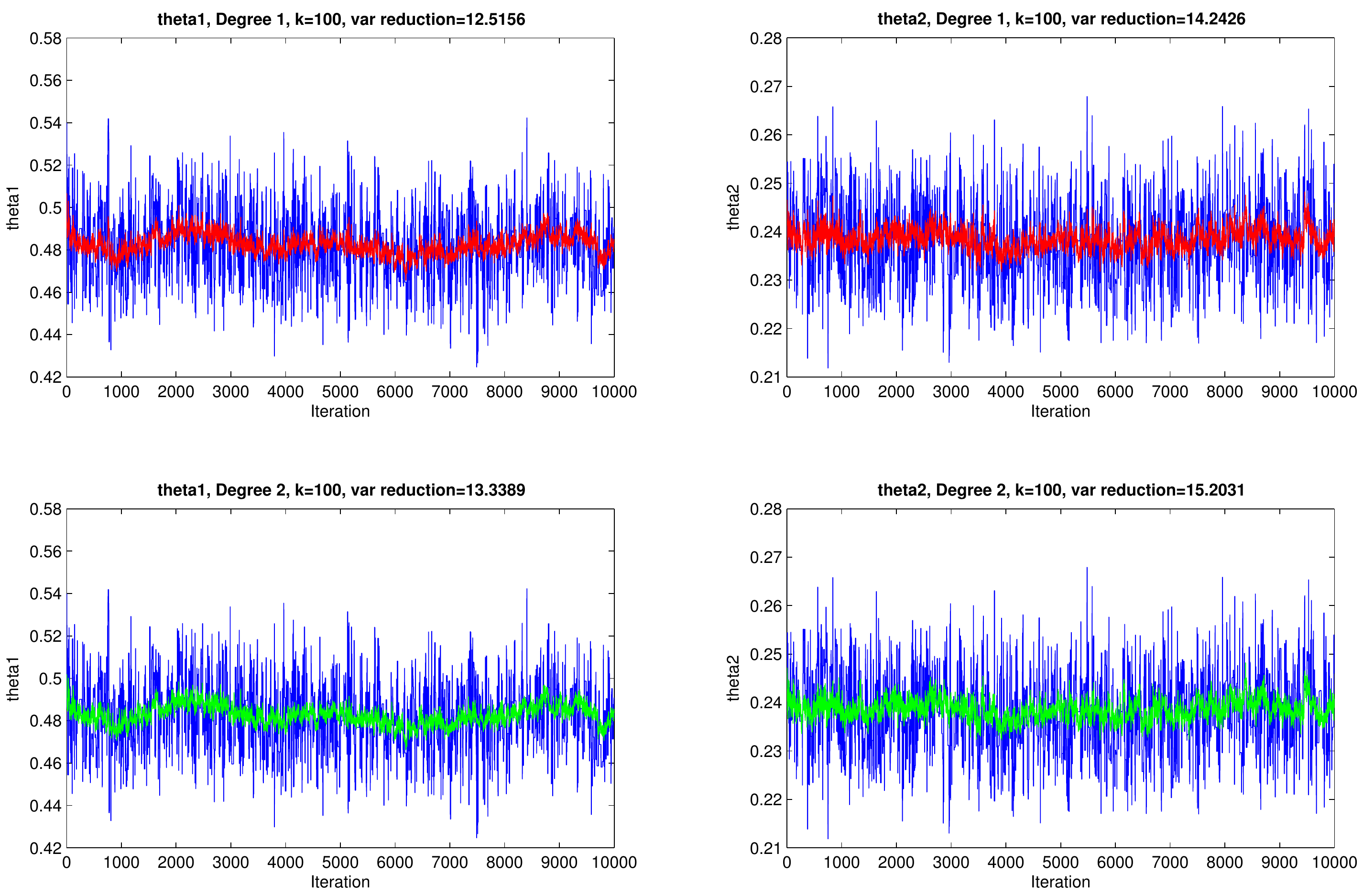}
\caption{SIR model. The top row displays the trace plots for $\theta_1$ and $\theta_2$ in uncontrolled (blue) and controlled (red) versions for a degree-one polynomial. The bottom row is similar, but for degree-two polynomials.}
\label{fig:SIR}
\end{figure}

Consider the specific example of the Susceptible-Infected-Recovered (SIR) model from epidemiology.
Letting $X_1$, $X_2$ denote respectively the proportions of susceptible and infected individuals in a population, modelled as continuous random variables, the SIR model has a stochastic representation given by
\begin{eqnarray}
\bm{\alpha}(\bm{X};\bm{\theta}) = \left[\begin{array}{c} -\theta_1 X_1X_2 \\ \theta_1 X_1X_2 - \theta_2 X_2 \end{array} \right], \; \; \; \; \; \bm{\beta}(\bm{X};\bm{\theta}) = \frac{1}{N} \left[\begin{array}{cc} \theta_1X_1X_2 & -\theta_1X_1X_2 \\ -\theta_1X_1X_2 & \theta_1X_1X_2 + \theta_2X_2 \end{array}\right]
\end{eqnarray}
where $N$ is a fixed population size and the rate parameters $\bm{\theta} \in [0,\infty)^2$ are unknown.
We assess our methodology by attempting to estimate the posterior mean of $\bm{\theta}$, taking $g(\bm{\theta}) = \theta_j$ for $j = 1,2$ in turn.
Here each $\theta_j$ was assigned an independent Gamma prior with shape and scale hyperparameters both equal to 2.
This prior vanishes at the origin and has exponentially decaying tails, so that the boundary condition of Lemma \ref{conditions} is satisfied by all polynomials.
Data were generated using the initial condition $\bm{X}_0 = [0.99,0.01]$, population size $N=1,000$ and parameters $\bm{\theta} = [0.5,0.25]$.
Observations were made at 20 evenly spaced intervals in the period from $t=0$ to $t=35$.
Five latent data points were introduced between each observed data point, so that the latent process has dimension $2 \times (20-1) \times 5 = 190$.
At each Monte Carlo iteration we sampled $K=100$ realisations of the latent data process $\bm{X}^u$.

Fig. \ref{fig:SIR} demonstrates that a variance reduction of about 12-14 times is possible using degree-one polynomials and 13-15 times using degree-two polynomials.
Again, these results highlight the potential to exploit multi-core processing for variance reduction in Monte Carlo methodology.

\section{Conclusions}

In this paper we have shown how repeated forward-simulation enables reduced-variance estimation in models that have intractable likelihoods.
The examples that we have considered illustrate the value of the proposed methodology in spatial statistics, social network analysis and inference for latent data models such as SDEs. 
The RV methodology provides a straight-forward means to leverage multi-core architectures for Bayesian estimation, that compliments recent work for MCMC in this direction by \cite{Alquier,Angelino,Bardenet,Calderhead,Korattikara,Maclaurin}.
Our theoretical analysis revealed that the number $K$ of forward-simulations should be taken equal to the number of cores in order to provide the optimal variance reduction per unit (serial) computation.
Furthermore, it was shown that the proposed
RV
estimator converges to the (intractable) ZV estimator of \cite{Mira} as the number $K$ of cores becomes infinite.
Our theoretical findings are supported by empirical results on standard benchmark datasets, that demonstrate a substantial variance reduction can be realised in practice.
In particular, results for the Ising model demonstrate that a 200-300 times variance reduction can be achieved by exploiting a $K=500$ core architecture.
More generally, our work shows that {\it it may not be necessary to parallelise the sampling process itself}; the potential of massively multi-core architectures can be harnessed in post-processing MCMC samples using control variates.

To conclude, we suggest interesting directions for further research:
\begin{itemize}
\item The approach that we pursued was a post-processing procedure that does not require modification to the MCMC sampling mechanism itself.
However an interesting possibility would be to also use the output of forward-sampling to construct gradient-based proposal mechanisms for the underlying MCMC sampler, following recent work in this direction by \cite{Alquier,Nemeth}.
This would retain the inherently parallel nature of the simulation procedure whilst yielding useful approximate Monte Carlo schemes that converge to an ``idealised'' (i.e. marginal) sampler as the number of forward-simulations, $K$, increases. In particular, our procedure can be implemented within the various schemes developed in \cite{Alquier,Nemeth} without any additional computational cost. 
Stochastic approximation of the score function was also recently considered by \cite{Atchade2} in the context of designing proximal gradient algorithms. Our work therefore combines to illustrate the wide range of statistical models for which such an approach could prove practically useful.
Alternative approaches to handling Type I intractability, such as Approximate Bayesian Computation \citep{Marjoram} typically also require a forward-simulation step and thus could also be embedded within our framework. 
\item In terms of statistical efficiency it would be interesting to extend the reduced-variance methodology to the non-parametric setting recently considered by \cite{Oates2}, that provides a mechanism to learn a suitable trial function $P(\bm{\theta})$ that need not be polynomial, leading (in some cases) to improved convergence rates.
A second interesting possibility would be to allow the number of forward-simulations, $K$ to depend upon the current state $\bm{\theta}$; in this way fewer simulations could be performed when it is expected that the score estimate $\hat{\bm{u}}(\bm{\theta}|\bm{y})$ is likely to have a low variance.
A third direction would be to move beyond independent estimation of the score $\bm{u}(\bm{\theta}|\bm{y})$ for each value of $\bm{\theta}$; here non-parametric regression techniques could play a role and this should yield further reductions in estimator variance, which again has a close analogy with \cite{Oates2}.

\item Finally, a referee suggested the intruiging possibility to expolit control variates for reduced-variance density estimation.
Specifically, we would take $g_{\bm{\theta}^*}(\bm{\theta}) = K_h(\bm{\theta}^*-\bm{\theta})$ for a bandwidth-$h$ kernel smoother estimate for the posterior density at a point $\bm{\theta}^* \in \Theta$, such that $\mathbb{E}_{\bm{\theta}|\bm{y}}[g_{\bm{\theta}^*}(\bm{\theta})] \rightarrow p(\bm{\theta}^*|\bm{y})$ as $h \rightarrow 0$. This is a direction that we are keen to explore further.
\end{itemize}

Of course, in the era of big data, as statisticians are increasingly interested in analysing larger datasets an immediate challenge is the issue of dealing with 
intractable likelihoods, due to the volume of data. Moreover, one would anticipate that statistical methodology will focus on the development of inferential algorithms that exploit modern 
multi-core computer architectures. Both of these will inevitably lead to further development and extensions of the methodology described in this paper. 

\bibliographystyle{ba}

\begin{acknowledgement}
The authors are grateful for the constructive feedback they received from the Editors and Reviewers at Bayesian Analysis.
NF was supported by the Science Foundation Ireland [12/IP/1424].
The Insight Centre for Data Analytics is supported by Science Foundation Ireland [SFI/12/RC/2289].
AM was supported by the Swiss National Science Foundation [CR12I1-156229].
CJO was supported by the EPSRC Centre for Research in Statistical Methodology [EP/D002060/1]. 
The authors thank Mark Girolami for helpful discussions.
\end{acknowledgement}

\section*{Appendix}

Below we prove Lemma \ref{one} from the Main Text. First we require a technical result:

\begin{lemma} \label{technical}
Write $\rho(\infty) = \text{\emph{Corr}}_{\bm{\theta}|\bm{y}}[g(\bm{\theta}),h(\bm{\theta}|\bm{y})]$.
There exists $C \in (0,\infty)$ such that
\begin{eqnarray}
\rho(K)^2 = \left( \frac{1}{\rho(\infty)^2} + \frac{C}{K} \right)^{-1}. \label{rho identity}
\end{eqnarray}
\end{lemma}
\begin{proof}
For Type I models write $\hat{h}(\bm{\theta},\bm{Y}) = h(\bm{\theta}) + \epsilon(\bm{\theta},\bm{Y})$ where $\bm{Y} = (\bm{Y}_1,\dots,\bm{Y}_K)$ and we suppress dependence on the data $\bm{y}$ in this notation.
It follows that the discrepancy between the reduced-variance and ZV control variates is given by
\begin{eqnarray}
\epsilon(\bm{\theta},\bm{Y}) = \nabla_{\bm{\theta}}P(\bm{\theta}) \cdot \left[ \frac{1}{K}\sum_{k=1}^K \bm{s}(\bm{Y}_k) - \mathbb{E}_{\bm{Y}'|\bm{\theta}}[\bm{s}(\bm{Y}')] \right].
\end{eqnarray}
Taking an analogous approach to Type II models we obtain
\begin{eqnarray}
\epsilon(\bm{\theta},\bm{X}) = \nabla_{\bm{\theta}}P(\bm{\theta}) \cdot \left[ \frac{1}{K}\sum_{k=1}^K \bm{u}(\bm{\theta},\bm{X}_k) - \mathbb{E}_{\bm{X}'|\bm{\theta},\bm{y}}[\bm{u}(\bm{\theta},\bm{X}')] \right].
\end{eqnarray}
Note that $\mathbb{E}_{\bm{Y}|\bm{\theta}}[\epsilon(\bm{\theta},\bm{Y})] = 0$ and hence 
\begin{eqnarray}
\mathbb{E}_{\bm{Y},\bm{\theta}|\bm{y}}[\epsilon(\bm{\theta},\bm{Y})] = \mathbb{E}_{\bm{Y}|\bm{\theta},\bm{y}}\mathbb{E}_{\bm{\theta}|\bm{y}}[\epsilon(\bm{\theta},\bm{Y})] = \mathbb{E}_{\bm{\theta}|\bm{y}} \mathbb{E}_{\bm{Y}|\bm{\theta}} [\epsilon(\bm{\theta},\bm{Y})] = 0, 
\end{eqnarray}
with an analogous result holding for Type II models.
Using these results we have that, for Type I models
\begin{eqnarray}
\mathbb{V}_{\bm{Y},\bm{\theta}|\bm{y}}[\epsilon(\bm{\theta},\bm{Y})] & = & \mathbb{E}_{\bm{Y},\bm{\theta}|\bm{y}}[\epsilon(\bm{\theta},\bm{Y})^2] \\
& = & \mathbb{E}_{\bm{\theta}|\bm{y}} \mathbb{E}_{\bm{Y}|\bm{\theta}}[\epsilon(\bm{\theta},\bm{Y})^2] = \mathbb{E}_{\bm{\theta}|\bm{y}} \mathbb{V}_{\bm{Y}|\bm{\theta}}[\epsilon(\bm{\theta},\bm{Y})] \\
& = & \mathbb{E}_{\bm{\theta}|\bm{y}}[K^{-1}\mathbb{V}_{\bm{Y}|\bm{\theta}}[\nabla_{\bm{\theta}} P(\bm{\theta}) \cdot \bm{s}(\bm{Y})]] = C_1 K^{-1}
\end{eqnarray}
where $C_1 = \mathbb{E}_{\bm{\theta}|\bm{y}} \mathbb{V}_{\bm{Y}|\bm{\theta}}[\nabla_{\bm{\theta}} P(\bm{\theta}) \cdot \bm{s}(\bm{Y})]$.
Also observe that, for any function $g(\bm{\theta})$, 
\begin{eqnarray}
\text{Cov}_{\bm{Y},\bm{\theta}|\bm{y}}[g(\bm{\theta}),\epsilon(\bm{\theta},\bm{Y})] & = & \mathbb{E}_{\bm{Y},\bm{\theta}|\bm{y}}[g(\bm{\theta})\epsilon(\bm{\theta},\bm{Y})] - \mathbb{E}_{\bm{Y},\bm{\theta}|\bm{y}}[g(\bm{\theta})]\mathbb{E}_{\bm{Y},\bm{\theta}|\bm{y}}[\epsilon(\bm{\theta},\bm{Y})] \\
& = & \mathbb{E}_{\bm{\theta}|\bm{y}} \mathbb{E}_{\bm{Y}|\bm{\theta}}[g(\bm{\theta})\epsilon(\bm{\theta},\bm{Y})] = 0.
\end{eqnarray}
Putting these results together we obtain
\begin{eqnarray}
\rho(K) = \frac{\text{Cov}_{\bm{Y},\bm{\theta}|\bm{y}}[g(\bm{\theta}),\hat{h}(\bm{\theta},\bm{Y})]}{\sqrt{\mathbb{V}_{\bm{Y},\bm{\theta}|\bm{y}}[g(\bm{\theta})]} \sqrt{\mathbb{V}_{\bm{Y},\bm{\theta}|\bm{y}}[\hat{h}(\bm{\theta},\bm{Y})]}} = \frac{\text{Cov}_{\bm{\theta}|\bm{y}}[g(\bm{\theta}),h(\bm{\theta})]}{\sqrt{\mathbb{V}_{\bm{\theta}|\bm{y}}[g(\bm{\theta})]} \sqrt{\mathbb{V}_{\bm{\theta}|\bm{y}}[h(\bm{\theta})] + C_1K^{-1}}}
\end{eqnarray}
from which it follows that 
\begin{eqnarray}
\frac{1}{\rho(K)^2} = \frac{\mathbb{V}_{\bm{\theta}|\bm{y}}[g(\bm{\theta})] (\mathbb{V}_{\bm{\theta}|\bm{\bm{Y}}}[h(\bm{\theta})] + C_1K^{-1})}{\text{Cov}_{\bm{\theta}|\bm{y}}[g(\bm{\theta}),h(\bm{\theta})]^2}
= \frac{1}{\rho(\infty)^2} + \frac{C}{K}
\end{eqnarray}
where $C = C_1 / \text{Cov}_{\bm{\theta}|\bm{y}}[g(\bm{\theta}),h(\bm{\theta})]^2$.
The analogous derivation for Type II models completes the proof.
\end{proof}

A simple corollary of Lemma \ref{technical} is that the reduced-variance estimator converges to the (unavailable) ZV estimator as $K \rightarrow \infty$.
Moreover we can derive an optimal choice for $K$ subject to fixed computational cost:

\begin{proof}[Proof of Lemma \ref{one}]
Starting from Eqn. \ref{tradeoff}, we substitute $I = c/ \ceil{K/K_0}$ and use the identity in Lemma \ref{technical} to obtain
\begin{eqnarray}
r(K) = \frac{1}{c} \ceil[\bigg]{\frac{K}{K_0}} \left( 1 - \frac{K \rho(\infty)^2}{K + C \rho(\infty)^2} \right). \label{cnvr}
\end{eqnarray}
Fig. S3 displays typical cost-normalised variance ratios $r(K)$, for both non-Gaussian (i.e. $\rho(\infty) < 1$; left) and Gaussian (i.e. $\rho(\infty) = 1$; right) cases.
In each case it the minimum is attained at $K = K_0$.
In general, we see from first principles that (i) $r > 0$, (ii) $r(k)$ is increasing for $k = K_0, 2K_0, 3K_0, \dots$, and (iii) $r(k)$ is decreasing on the interval $((n-1)K_0,nK_0]$ and bounded below by $r((n-1)K_0)$, for any $n \in \mathbb{N}$. Thus we have $\arg\min_{K = 1,2,3,\dots} r(K) = K_0$, as required.
\end{proof}

\end{document}